\documentclass[12pt]{article}

    \usepackage{subcaption} 
\usepackage{booktabs} 
\usepackage{tabularx}
\def\spacingset#1{\renewcommand{\baselinestretch}%
{#1}\small\normalsize} \spacingset{1}

\usepackage[linesnumbered,ruled,vlined]{algorithm2e}
\usepackage{graphicx,xcolor}
\usepackage{amsmath,amsthm,amssymb,bm}
\usepackage{natbib}
\usepackage{enumitem}
\usepackage{multirow}
\usepackage{placeins}
\usepackage{authblk}

\DeclareMathOperator*{\argmin}{argmin}

\newtheorem{assumption}{Assumption}

\newcommand{\suit}[1]{\left( #1\right)}
\newcommand{\msuit}[1]{\left[ #1\right]}
\newcommand{\set}[1]{\left\{ #1\right\}}
\newcommand{\abs}[1]{\left| #1\right|}

\newcommand{\mbX}{\boldsymbol{X}}
\newcommand{\mbx}{\boldsymbol{x}}

\newcommand{\mbu}{\boldsymbol{u}}

\newcommand{\mbv}{\boldsymbol{v}}
\newcommand{\bolbeta}{\boldsymbol{\beta}}

\newcommand{\bE}{\mathbb{E}}
\newcommand{\mI}{\mathcal{I}}
\newcommand{\mL}{\mathcal{L}}
\newcommand{\bolSigma}{\boldsymbol{\Sigma}}

\theoremstyle{plain}
\newtheorem{lemma}{Lemma}

\newtheorem{theorem}{Theorem}
\newtheorem{corollary}{Corollary}
\newtheorem{remark}{Remark}

\theoremstyle{remark}

\title{Generalized Linear Models for Extremes: Estimation and Inference in High Dimensions}

\author[1]{Liujun Chen}
\author[2]{Chen Zhou}

\affil[1]{School of Management, University of Science and Technology of China}
\affil[2]{Econometric Institute, Erasmus University Rotterdam} 
\begin{document}
 
\maketitle

\begin{abstract}
We propose a regression model for the extreme tail of a response variable, in which covariates rescale the tail without changing its shape. A single covariate-dependent function then characterizes the entire conditional tail, in contrast to extreme quantile regression, which targets a quantile at a pre-specified level. The tail shape itself is left unrestricted: heavy-, light- and short-tailed responses are covered by the same framework. We specify the function through a link function and a linear combination of the covariates, which is
in the spirit of a generalized linear model. In estimation, we match the parametric specification to the underlying tail function under a Bregman divergence, over a region localized at the largest observations. The resulting loss is convex, and an $\ell_1$-penalty allows the number of covariates to exceed the effective sample size. 
The tail localization makes the asymptotic theory deviate from that for 
classical penalized generalized linear models. Only the tail 
observations selected by a random threshold are used in the statistical 
analysis, making them dependent.
We derive the convergence rate of the penalized estimator and propose a debiased estimator that is asymptotically normal, yielding confidence intervals for individual coefficients. Its asymptotic variance is determined by the covariance of the score, which under tail localization differs from the Hessian and must be estimated separately. We apply the method to automobile insurance claims data.
\end{abstract}

\spacingset{1.9} 

\clearpage
\section{Introduction}
 
In many applications the quantity of interest is not the average behavior of a response but its behavior in the tail. The covariates that shift the center of a distribution need not be the ones that shift its extreme. An automobile insurer, for instance, cares less about the mean claim than about which policyholders generate the largest claims. The same distinction appears in other contexts. Adverse macroeconomic conditions can raise the probability of an extreme financial loss while barely moving median returns. A climate covariate can change how often extreme precipitation occurs without changing typical rainfall. Regression methods that describe the conditional mean, or conditional quantiles at moderate levels, carry little information about extreme events.

A second difficulty is that covariate information is often high-dimensional. The question is not only how covariates affect the tail but which among many candidates do. The two difficulties compound. In tail estimation the data are scarce by construction, since only the largest observations are informative, whereas variable selection requires large samples.

Extreme quantile regression \citep{chernozhukov2005extremal, wang2012estimation} is a natural tool for such questions. There, quantile regression is applied at quantile levels tending to one, usually with a Pareto-type extrapolation beyond the range of the data; see \cite{tang2026recent} for a survey of this literature. Recent studies on extreme quantile regression allow for non-linearity in the covariate effect, obtained by kernel smoothing \citep{daouia2023inference} or by machine learning methods such as gradient boosting \citep{velthoen2023gradient}. \cite{tang2024high} propose a regularized framework for extreme conditional quantiles with high-dimensional covariates, under a linear specification and without inference on the coefficients. By contrast, \cite{leng2026reinforced} obtain inference theory through a double bootstrap but stay in fixed dimension.

Two gaps in the literature motivate this study. First, these methods target the conditional quantile at a pre-specified level. They describe one point of the conditional tail rather than the tail itself. Second, most of them require the response variable to be heavy-tailed. We take a different route by assuming that covariates rescale the tail without reshaping it.

Our model follows the framework of heteroscedastic extremes, in which the tail of the response variable $Y$ varies systematically across the covariate space. Let $\mbX =(X_1,\dots, X_p)^\top \in \mathbb{R}^{p} $ denote the vector of covariates, where, following standard practice, we set $X_1 \equiv 1$ to incorporate an intercept. Specifically, we assume that there exists a non-negative continuous function $c: \mathbb{R}^p \to \mathbb{R}^+$ such that
\begin{equation}\label{eq:model}
	\lim_{y\to y^{+}} \frac{1-F_Y(y|\mbX=\mbx)}{1-F_Y(y)} =\lim_{y\to y^{+}} c_y(\mbx) = c(\mbx),
\end{equation}
where $y^{+}$ denotes the upper endpoint of the marginal support of $Y$. Under this formulation, the conditional tail of $Y$ given $\mbX=\mbx$ is asymptotically proportional to the marginal tail, with the function $c(\mbx)$ quantifying the relative tail intensity associated with the covariate profile $\mbx$.

The model implies that the conditional and marginal distributions of $Y$ share the same upper endpoint and the same extreme value index. In other words, covariates move the scale of the tail, not its shape. We adopt this restriction deliberately. The extreme value index is the hardest tail feature to estimate, and letting it vary with covariates may lead to complications in statistical inference; see \cite{yang2025cautions} for a discussion. In return, the function $c$ pins down the entire conditional tail instead of a single extreme quantile. In addition, the model \eqref{eq:model} does not restrict the common index itself, which may be positive, zero or negative. The framework thus covers heavy-tailed responses along with light- and short-tailed ones; see Section \ref{sec:addition:simulation} of the Supplementary Material   for a short-tailed design.

\cite{einmahl2016statistics} introduced the heteroscedastic extremes framework and called $c$ the scedasis function, estimating it nonparametrically for a scalar index. \cite{mefleh2020trend} and \cite{einmahl2022spatial} extended it to trend detection and to space-time data. Throughout that literature the covariate is one-dimensional, typically time or location. Neither high-dimensional covariates nor inference on covariate effects is available. Sparsity in extremes has been pursued from other directions, though not for the scedasis function; see the survey in \cite{engelke2021sparse}.

In this paper, we adopt a flexible parametric specification for the heteroscedastic function $c$, 
$$
c_{\bolbeta}(\mbx):= g\suit{\mbx^\top\bolbeta }, \quad \bolbeta = (\beta_1,\dots,\beta_p)^\top \in \mathbb{R}^p, 
$$
where $g(\cdot)$ is a pre-specified positive link function. The specification combines a link function with a linear combination of the covariates, as in a generalized linear model. Taking the expectation on both sides of \eqref{eq:model} yields the normalization constraint
\begin{equation}\label{eq:condition:identify}
	\bE \suit{c_{\bolbeta}(\mbX)}=1.  
\end{equation}
The constraint is a consequence of the model, and it identifies $\bolbeta$.

We estimate $\bolbeta$ by minimizing a Bregman divergence between the conditional tail distribution and its parametric counterpart, over a region localized at the $k$ largest responses. A fundamental challenge in deriving the asymptotic theory arises from the scarcity of tail observations. Although the available sample size is $n$, the information relevant to tail inference is primarily contained in the largest $k$ observations. Consequently, the effective sample size is governed by $k$ rather than $n$, substantially amplifying the difficulty of high-dimensional estimation. To address this issue, we incorporate an $\ell_1$-penalty and establish the convergence rate $\sqrt{s\log p/k}$ under the scaling condition
$
s\log p =o(k),
$
where $s$ denotes the sparsity level. The condition permits the number of covariates to exceed the effective sample size $k$. This contrasts with classical high-dimensional theory, where the analogous condition is $s\log p = o(n)$. 

A second complication is that tail observations are selected through a common random order-statistic threshold. The selected observations are therefore dependent, even when the original observations are independent and identically distributed (i.i.d.). This dependence prevents a direct application of standard generalized linear model arguments. 

To conduct valid inference for regression coefficients, we introduce a link-specific debiasing scheme based on sample-splitting \citep{cai2023statistical}. We establish the asymptotic normality of the proposed debiased estimator. The standard information matrix identity breaks down in the localized tail region. Consequently, the covariance of the score and the Hessian are no longer asymptotically equivalent, and the covariance must be estimated separately. Our debiasing procedure constructs the projection direction by minimizing the estimated covariance of the score, subject to constraints based on the Hessian.

To our knowledge, this is the first treatment of the scedasis function with high-dimensional covariates and the first to deliver coefficient-level inference in that model. The methodology in this study connects extreme value analysis to high-dimensional generalized linear models, a well-developed area \citep{van2008high,negahban2012unified,guo2016tests,shi2019linear,cai2023statistical}. Our method adapts the estimation of generalized linear models to a region localized at the tail.

The rest of the paper is organized as follows. Section \ref{sec:estimation} presents the penalized estimation strategy and its corresponding theory. Section \ref{sec:inference} considers the debiasing procedure and statistical inference for the estimator. A simulation study is presented in Section \ref{sec:simulation}, and a real data application is provided in Section \ref{sec:realdata}. A concluding discussion is provided in Section \ref{sec:discussion}.
 All technical proofs are collected in the Supplementary Material. 

We define the following notation. For any vector $\mbv = (v_1,\dots, v_p)^\top \in \mathbb{R}^p$, $\|\mbv\|_0 = \sum_{j=1}^p \mI\suit{v_j \ne 0 } $, $\|\mbv\|_1 = \sum_{j=1}^p |v_j|$, $\|\mbv\|_2 = \suit{\sum_{j=1}^p v_j^2}^{1/2}$, $\|\mbv\|_{\infty} = \max_{1\le j\le p} |v_j|$. For any square matrix $\boldsymbol{A}$, $\lambda_{\min}(\boldsymbol{A})$ and $\lambda_{\max}(\boldsymbol{A})$ denote the least and largest eigenvalues of $\boldsymbol{A}$, respectively. We use $\nabla$ to denote the gradient operator, and $\nabla^2$ to denote the Hessian. For any real number $x\ne 0$, $\text{sgn}(x)$ denotes the sign of $x$. 
 A mean-zero random variable $X$ is sub-Gaussian with parameter $\sigma$ if its moment generating function satisfies 
 $$
 \bE \exp\suit{ sX} \le \exp\suit{\sigma^2 s^2/2}.
 $$	
For two positive sequences, $a_n \asymp b_n$ means that $a_n/b_n = O(1)$ and $b_n/a_n=O(1)$ as $n\to\infty$.

\section{Estimation}\label{sec:estimation}

\subsection{Penalized Estimation}
 
We estimate $\bolbeta$ by matching the parametric model $c_{\bolbeta}$ to the true heteroscedastic function $c_y$ under a Bregman divergence. Let $\phi: \mathbb{R}_{+}\to \mathbb{R} $ be a strictly convex, differentiable function. The Bregman divergence between $c_y$ and $c_{\bolbeta}$, weighted by the covariate distribution, is 
$$
\mathcal{D}_{\phi}(c_y, c_{\bolbeta}) = \int \msuit{ \phi(c_y(\mbx))- \phi(c_{\bolbeta}(\mbx)) - \phi^\prime (c_{\bolbeta}(\mbx) )\set{c_y(\mbx)-c_{\bolbeta}(\mbx)} }f_{\mbX}(\mbx)d\mbx, 	
$$
where $f_{\mbX}$ is the probability density function of the covariates $\mbX$. 
Ignoring terms that do not depend on $\bolbeta$, minimizing 
$\mathcal{D}_{\phi}(c_y, c_{\bolbeta})$ is equivalent to minimizing
\begin{align*}
\mL_y(\bolbeta) = \int \set{ c_{\bolbeta}(\mbx) \phi^\prime (c_{\bolbeta}(\mbx)) - \phi(c_{\bolbeta}(\mbx)) }f_{\mbX}(\mbx)d\mbx - \int \phi^\prime (c_{\bolbeta}(\mbx)) c_y(\mbx) 	f_{\mbX}(\mbx)d\mbx.
\end{align*}
By Bayes' rule,
$$
c_y(\mbx)= \frac{1-F_Y(y|\mbX=\mbx)}{1-F_Y(y)} = \frac{f_{\mbX}(\mbx|Y>y)}{ f_{\mbX}(\mbx)}.
$$
Consequently,   
\begin{align*}
	\mL_y(\bolbeta) =& \int \set{  c_{\bolbeta}(\mbx) \phi^\prime (c_{\bolbeta}(\mbx)) - \phi(c_{\bolbeta}(\mbx))   }f_{\mbX}(\mbx)d\mbx - \int \phi^\prime (c_{\bolbeta}(\mbx))   	f_{\mbX}(\mbx|Y>y)d\mbx \\
	=& \int \Psi(c_{\bolbeta}(\mbx)) f_{\mbX}(\mbx)d\mbx - \int \phi^\prime (c_{\bolbeta}(\mbx))   	f_{\mbX}(\mbx|Y>y)d\mbx\\
	=& \int \Psi(g(\mbx^\top\bolbeta )) f_{\mbX}(\mbx)d\mbx - \int \phi^\prime (g(\mbx^\top\bolbeta ))   	f_{\mbX}(\mbx|Y>y)d\mbx,
\end{align*}
where $\Psi(u) = u\phi^\prime(u) - \phi(u).$

We focus on the generator $\phi$ satisfying 
$$
  \phi^\prime(u) = g^{-1}(u).
$$
Under this specification, we have the identity
  $$
   \phi^\prime (g(\mbx^\top\bolbeta )) = \mbx^\top \bolbeta.
  $$
  Furthermore, by applying the chain rule in conjunction with the inverse function theorem, it follows that $d \Psi(g(u))/du = g(u)$, and hence 
  $$
  \Psi(g(u)) = \int_0^u g(t) dt +C =: G(u)+C,  
  $$
  for some constant $C$. 
Therefore, the resulting loss function simplifies to (up to a constant that does not depend on $\bolbeta$)
 $$
 \mL_y(\bolbeta)= \int G(\mbx^\top\bolbeta) f_{\mbX}(\mbx)d\mbx - \int \suit{\mbx^\top\bolbeta} f_{\mbX}(\mbx|Y>y)d\mbx.
 $$
 
 Note that
 \begin{align*}
 	 \frac{\partial \mL_y(\bolbeta) }{\partial \beta_1} =& \int  g\suit{\mbx^\top\bolbeta} f_{\mbX}(\mbx)d\mbx -\int    f_{\mbX}(\mbx|Y>y)d\mbx \\
 	 =&\int   g\suit{\mbx^\top\bolbeta} f_{\mbX}(\mbx)d\mbx - 1 \\
 	 =& \bE c_{\bolbeta}(\mbX)  -1.
 \end{align*}
Therefore, any minimizer of $ \mL_y(\bolbeta)$ must satisfy the normalization constraint \eqref{eq:condition:identify}.

Assume that we observe i.i.d. observations $(\mbX_i, Y_i)$, $i=1,\dots,n$.
 Let $Y_{n-k,n}$ denote the $k+1$-th largest order statistic of $\set{Y_1,\dots, Y_n}$, where $k$ is an intermediate sequence satisfying $k = k(n)\to \infty, k/n\to 0$ as $n\to\infty$. Replacing expectations by empirical averages yields the sample loss 
$$
\mL_n(\bolbeta) = \frac{1}{n}\sum_{i=1}^n G\suit{\mbX_i^\top \bolbeta} - \frac{1}{k}\sum_{i=1}^n \mbX_i^\top \bolbeta \mI\suit{Y_i>Y_{n-k,n}}.
$$
Structurally, $\mathcal{L}_n(\boldsymbol{\beta})$ adapts a generalized-linear-model loss to localized tail inference. The first term, $n^{-1} \sum_{i=1}^n G(\mathbf{X}_i^\top \boldsymbol{\beta})$, is a response-independent convex term computed from all covariates. The second term uses the covariate average among observations above the random threshold $Y_{n-k,n}$. Thus, the response enters through the $k$ upper-order observations, whereas the full covariate sample contributes to normalization and curvature.

To accommodate high-dimensional covariates, we consider the $\ell_1$-penalized estimator
 $$
 \widehat{\bolbeta}_n = \argmin_{\bolbeta} \mL_n(\bolbeta) +\lambda_n \|\bolbeta\|_1.
 $$

\subsection{Theory for fixed dimensions}

 In this section, we study the asymptotic behavior of $\widehat{\bolbeta}_n$ when $p$ is fixed. We impose the following conditions. Let $\bolbeta^0$ denote the target parameter.
 
\renewcommand{\theassumption}{(M)}
\begin{assumption}\label{all:model:bias}
There exist functions $A$ and $B$, with $A(y)\to 0$ as $y\to y^{+}$, such that
	\begin{equation*}
 \sup_{\mbx \in \mathcal{X}} \frac{1}{B(\mbx)} \abs{\frac{1-F_Y(y|\mbX=\mbx)}{1-F_Y(y)}- g\suit{\mbx^\top\bolbeta^0 }} = O(1)A(y),
\end{equation*}
where $\mathcal{X}$ is the support of the covariates $\mbX$.  
 Moreover, for some $\delta>0$,  
$$
 \bE |B(\mbX)|^{1+\delta}<\infty, \quad \max_{1\le j\le p} \bE \suit{|X_j B(\mbX)|} < \infty.
$$	
\end{assumption}
Assumption \ref{all:model:bias} quantifies the rate of convergence of the tail approximation error in model \eqref{eq:model}. Similar conditions are commonly used in extreme value statistics to control model bias \citep{haan2006extreme, einmahl2016statistics, xu2022prediction}. 
 The function $B(\mathbf{x})$ serves as a spatial scaling factor for the second-order expansion. When the tail approximation error is uniformly bounded over the covariate support $\mathcal{X}$, one can naturally specify $B(\mathbf{x}) \equiv 1$. In this case, the regularized moment conditions reduce directly to the standard first-moment constraint on the individual features, requiring only that $\max_{1 \le j \le p} \mathbb{E}[|X_j|] < \infty$.

\setcounter{assumption}{0}   
\renewcommand{\theassumption}{(A\arabic{assumption})}
\begin{assumption}\label{fix:condition:empiricalprocess}
Let $U = 1 -F_Y(Y)$. For any $j \in \set{1,\dots,p}$, the functions $u\mapsto \bE(X_j \mI\suit{U\le u})$ and $u\mapsto \bE(X_j^2 \mI\suit{U\le u})$ are differentiable on $(0,1)$ with continuous derivatives at $0$.
\end{assumption}

\begin{assumption}\label{fix:condition:moment:bound}
	For some $\delta, \ \varepsilon>0$, for any $j \in \set{1,\dots,p}$,
$$
  \bE \suit{ \suit{X_jg(\mbX^\top \bolbeta^0)}^{2} }<\infty, \quad \bE \suit{\sup_{\|\bolbeta -\bolbeta^0 \|_2 \le \varepsilon} \abs{X_j^2g^\prime\suit{ \mbX^\top \bolbeta} }^{(1+\delta)} }<\infty.
$$
\end{assumption}
Assumption \ref{fix:condition:empiricalprocess} is a technical regularity condition on the smoothness of the functions $\bE(X_j \mI\suit{U\le u})$ and $\bE(X_j^2 \mI\suit{U\le u})$. The same condition is imposed in \cite{aghbalou2024tail}. 
Assumption \ref{fix:condition:moment:bound} is a regularity condition on the moments of the covariates and the link function. The assumption is mild and is satisfied in many practical settings. For instance, if $\mbX$ has compact support and $g'$ is continuous, then the required moment bounds hold automatically for sufficiently small $\varepsilon$.

\begin{theorem}\label{fix:theorem:normality}
Suppose that Assumptions \ref{all:model:bias}, \ref{fix:condition:empiricalprocess} and \ref{fix:condition:moment:bound} 	hold. Moreover, assume that as $n\to\infty$, $k\to\infty, k/n\to 0$, $\sqrt{k}A(F_Y^{-1}(1-k/n))\to 0$ and $\sqrt{k}\lambda_n\to \lambda_0$, where $F_Y^{-1}$ denotes the quantile function of $Y$. Then, as $n\to\infty$, 
$$
\sqrt{k}\suit{\widehat{\bolbeta}_n - \bolbeta^0 } \stackrel{d}{\to} \argmin_{\Delta} \suit{ -\Delta^\top \boldsymbol{W}+\frac{1}{2}\Delta^\top \boldsymbol{H}(\bolbeta^0)\Delta+ b(\Delta) }. 
$$
Here, $b(\Delta) = \lambda_0\sum_{j=1}^p \suit{\Delta_j\text{sgn}(\beta_j^0)\mI\suit{\beta_j^0\ne 0}+|\Delta_j|\mI(\beta_j^0=0) }$, $\boldsymbol{W}$ is a mean-zero multivariate Gaussian random vector with covariance matrix 
 $$
 \boldsymbol{\Sigma}_{\boldsymbol{W}} = \bE \suit{ \mbX\mbX^\top g(\mbX^\top \bolbeta^0) } - \bE \suit{ \mbX g(\mbX^\top \bolbeta^0) } \bE \suit{ \mbX^\top g(\mbX^\top \bolbeta^0) },
 $$
 and 
 	$$
	\boldsymbol{H}(\bolbeta^0) = \bE \suit{ \mbX\mbX^\top g^\prime\suit{\mbX^\top\bolbeta^0}}. 
	  	$$
 \end{theorem}
 \begin{corollary}
 Under the assumptions of Theorem \ref{fix:theorem:normality} and $\lambda_0=0$, the proposed estimator is asymptotically normal; that is, as $n\to\infty$,
 $$
 \sqrt{k}\suit{\widehat{\bolbeta}_n - \bolbeta^0 } \stackrel{d}{\to} N\suit{\boldsymbol{0}, \set{\boldsymbol{H}(\bolbeta^0)}^{-1} \boldsymbol{\Sigma}_{\boldsymbol{W}} \set{\boldsymbol{H}(\bolbeta^0)}^{-1} }.
 $$
 \end{corollary}

 \subsection{Theory in high dimensions}
 \setcounter{assumption}{0}   
\renewcommand{\theassumption}{(B\arabic{assumption})}

 \begin{assumption}\label{condition:regular:x} 
 \
 
 \begin{itemize}
 	\item[(i)] The covariates $\mbX$ are from a mean-zero distribution with covariance matrix $\bolSigma = \textnormal{Cov}(\mbX)$ such that $0<\kappa_{l}\le \lambda_{\min}(\boldsymbol{\Sigma})<\lambda_{\max}(\bolSigma)\le \kappa_s<\infty $.
 	\item[(ii)] The covariates $X_j, j=1,\dots, p$ are bounded, i.e., $ \max_{1\le j\le p} |X_j|\le \kappa_X$.  
	\item[(iii)] For any $\boldsymbol{v}\in \mathbb{R}^p$, the variable $\mbX^\top \boldsymbol{v}$ is sub-Gaussian with parameter at most $\kappa_u \|\boldsymbol{v}\|_2$.  
	\item[(iv)] There exists a constant $R>0$ such that $\|\bolbeta^0\|_2\le R$. 
 \end{itemize}
\end{assumption}
 
Assumption \ref{condition:regular:x} imposes standard regularity conditions on the design distribution. Such conditions are commonly assumed in high-dimensional regression, see for example \cite{buhlmann2011statistics, negahban2012unified}.   
 The boundedness of $\|\bolbeta^0\|_2$ guarantees that the linear functional $\mathbf{X}^\top \boldsymbol{\beta}^0$ remains uniformly well-behaved and does not introduce additional growth in high dimensions.

\begin{assumption}\label{high:condition:moment} 
\

\noindent (i)
The function $g^\prime$ is continuous, positive and increasing. 

\noindent (ii) There exist constants $R>0, \delta>0$ such that $\|\bolbeta^0\|_2\le R$, and 
	$$
 \bE g^2(\mbX^\top\bolbeta^0)\le R, \quad \bE \abs{g^\prime\suit{ \mbX^\top \bolbeta^0}}^{1+\delta } \le R.
$$

\noindent (iii) Moreover, for some constant $T>\kappa_uR$, $g(T\sqrt{\log n}) \sqrt{\log p/n}\to 0$ as $n\to\infty$.
\end{assumption}

The conditions in Assumption \ref{high:condition:moment} impose mild restrictions on the link function $g$ and the dimensionality $p$, which are easily satisfied by widely used link functions. For instance, when $g(x) = \exp(x)$, condition (i) holds inherently, and condition (ii) is satisfied under sub-Gaussian covariates. For condition (iii), $g(T\sqrt{\log n})/n^{\delta}\to 0$ holds for any $\delta>0$, implying that a sufficient condition is simply $\log p/n^{1-\delta}\to 0$ for some $\delta>0$. When $g(x)=\log(1+e^x)$ (Softplus), its derivative $g^\prime(x)$ is bounded and monotonic, rendering conditions (i) and (ii) straightforward under sub-Gaussian covariates, while condition (iii) accommodates a much larger $p$ due to the linear growth of $g$.

\begin{theorem}\label{theorem:high:main}
Suppose that Assumptions \ref{all:model:bias}, \ref{condition:regular:x},
  and \ref{high:condition:moment} hold. Assume that as $n\to\infty$, $k\to\infty$ and $k/n\to 0$, and that
$
\sqrt{k}\, A\!\left(F_Y^{-1}(1-k/n)\right) = o\!\left(\sqrt{\log p}\right).
$
Let $s := \|\bolbeta^0\|_0$, and suppose that $s\log p = o(k)$.
Let $
\lambda_n = c_0 \sqrt{\log p/k},
$
for some constant $c_0>0$. Then, with probability tending to 1, 
$$
\|\widehat{\bolbeta}_n - \bolbeta^0\|_2
\lesssim \sqrt{ \frac{s\log p}{k}}.
$$
\end{theorem}
Theorem \ref{theorem:high:main} establishes the convergence rate for the proposed $\ell_1$-regularized estimator. Compared with the existing literature on generalized linear models \citep{negahban2012unified, cai2023statistical}, the convergence rate is of order $\sqrt{ s\log p/k}$ rather than $\sqrt{ s\log p/n}$. This difference arises because the effective sample size in our setting is $k$ rather than $n$. More specifically, although the observed dataset contains $n$ samples, the information available for estimating the parameter is determined by the subsample of size $k$ that contributes to the estimation procedure.

 \section{Debiasing and Inference}\label{sec:inference}
The regularized estimator $\widehat{\bolbeta}_n$
exhibits a non-negligible bias relative to its variance, precluding direct asymptotic inference. In high-dimensional settings, this issue is commonly addressed by constructing a debiased estimator to remove the bias \citep{van2014asymptotically, zhang2014confidence, javanmard2014confidence, ning2017general, cai2023statistical}. Motivated by this approach, we propose a debiased estimation and inference procedure in this section. Specifically, we employ a sample-splitting strategy in which the bias correction is carried out on a complementary data split that is independent of the subsample used to obtain the initial estimator \citep{cai2023statistical}.

We split the sample so that the initial estimation and bias correction are conducted on different samples. Without loss of generality, we assume there are $2n$ samples $\mathcal{D} =\set{(\mbX_i, Y_i)}_{i=1}^{2n}$, divided into two disjoint subsets $\mathcal{D}_1 =\set{(\mbX_i, Y_i)}_{i=1}^{n}$ and $\mathcal{D}_2 =\set{(\mbX_i, Y_i)}_{i=n+1}^{2n}$. The initial estimator $\widehat{\bolbeta}_n$ is computed from $\mathcal{D}_2$, and the score, projection direction, and bias correction are computed exclusively from the independent sample $\mathcal{D}_1$.  

For $j=2,\dots, p$, 
 we consider the following form of the bias-correction estimator 
 $$
 \widetilde{\beta}_{j} = \widehat{\beta}_{j} -\mbu_j^\top \mL_n^\prime (\widehat{\bolbeta}_n ), 
 $$
 where $\mbu_j \in \mathbb{R}^{p}$ is a weight vector. Define 
 $$
 \boldsymbol{\Sigma}_n(\bolbeta) = \frac{1}{n}\sum_{i=1}^{n} \mbX_i\mbX_i^\top g (\mbX_i^\top \bolbeta) - \suit{ \frac{1}{n}\sum_{i=1}^{n} \mbX_i g (\mbX_i^\top \bolbeta) }\suit{ \frac{1}{n}\sum_{i=1}^{n} \mbX_i^\top g (\mbX_i^\top \bolbeta) } . 
 $$

We construct the weight vector as follows:
\begin{align}
\widehat{\mbu}_j = &\argmin_{\mbu\in\mathbb{R}^p} \mbu^\top \boldsymbol{\Sigma}_n(\widehat{\bolbeta}_n)  \mbu , \\
	s.t.\quad  & \left \|\suit{ \frac{1}{n} \sum_{i=1}^n \mbX_i\mbX_i^\top g^\prime(\mbX_i^\top \widehat{\bolbeta}_n) } \mbu - e_j  \right\|_{\infty} \le \lambda_{1,n}, \label{bias:basis} \\
	&   \max_{1\le i\le n} |\mbX_i^\top \mbu| \le \lambda_{2,n}, \label{bias:max}   
\end{align}
where $\set{e_j}_{j=1}^p$ denotes the canonical basis of the Euclidean space $\mathbb{R}^p$, and $\lambda_{1,n}, \lambda_{2,n}$ satisfy the
conditions stated in Assumption \ref{condition:choice:weight} below.

The optimization maintains the classical debiasing principle, i.e., minimizing the variance of the projection under a bias-vanishing constraint. 
However, it fundamentally differs from standard GLMs in its variance structure due to our tail-inference framework. To see this, note that on the inference sample $\mathcal{D}_1$, the score function and the Hessian matrix evaluated at $\widehat{\bolbeta}_n$ are given by
$$
\begin{aligned}
	\mL_n^\prime(\widehat{\bolbeta}_n ) = & \frac{1}{n}\sum_{i=1}^n \mbX_i g(\mbX_i^\top \widehat{\bolbeta}_n) - \frac{1}{k}\sum_{i=1}^n \mbX_i \mI(Y_i>Y_{n-k,n}), \quad \text{and}\\
	\mL_n^{\prime \prime} (\widehat{\bolbeta}_n ) = & \frac{1}{n}\sum_{i=1}^n \mbX_i\mbX_i^\top  g^\prime(\mbX_i^\top \widehat{\bolbeta}_n),
\end{aligned}
$$ 
respectively. In classical GLMs, the information matrix identity equates the asymptotic covariance of the score with the Hessian matrix, allowing the objective function and the constraint to share the same matrix profile. In contrast, this identity breaks down in our framework due to the presence of the extreme-value indicator $\mI(Y_i>Y_{n-k,n})$. Consequently, while the constraint must still rely on the Hessian $\mL_n^{\prime \prime} (\widehat{\bolbeta}_n )$ to eliminate the regularized estimation bias, the objective function must minimize the true covariance $\boldsymbol{\Sigma}_n(\widehat{\boldsymbol{\beta}}_n)$. 

 \begin{remark}\label{remark:exp:rankone}
Write
$$
\boldsymbol{M}_n(\bolbeta) = \frac{1}{n}\sum_{i=1}^n \mbX_i\mbX_i^\top g\suit{\mbX_i^\top \bolbeta}, \qquad
\boldsymbol{b}_n(\bolbeta) = \frac{1}{n}\sum_{i=1}^n \mbX_i\, g\suit{\mbX_i^\top \bolbeta},
$$
so that $\boldsymbol{\Sigma}_n = \boldsymbol{M}_n - \boldsymbol{b}_n \boldsymbol{b}_n^\top$. Since $X_{i1}\equiv 1$, we have $\boldsymbol{M}_n e_1 = \boldsymbol{b}_n$ for every $\bolbeta$ and every link, and therefore
$$
\boldsymbol{\Sigma}_n = \boldsymbol{M}_n - \boldsymbol{M}_n e_1 e_1^\top \boldsymbol{M}_n.
$$
If in addition $g=g^\prime$, that is $g(x)=e^x$, then $\boldsymbol{M}_n = \mL_n^{\prime\prime}$ and the sandwich collapses to
$$
\suit{\mL_n^{\prime\prime}}^{-1}\boldsymbol{\Sigma}_n\suit{\mL_n^{\prime\prime}}^{-1} = \suit{\mL_n^{\prime\prime}}^{-1} - e_1 e_1^\top.
$$ 
The centering correction reduces the variance of the intercept by exactly one and leaves the variance of every $\beta_j$ with $j\ge 2$ unchanged. In other words, under the canonical exponential link a statistician who ignored the failure of the information identity and studentized with the Hessian would still obtain asymptotically valid intervals for the covariate coefficients. However, for any non-canonical link function, e.g. $g(x)=\log(1+e^x)$, one has $\boldsymbol{M}_n \ne \mL_n^{\prime\prime}$. The two studentizers differ at first order. Section \ref{sec:simu:falsify} reports what the difference costs.
\end{remark}

 Finally, by substituting the optimized weight vector $\widehat{\mathbf{u}}_j$, the finalized debiased estimator for each coordinate $j=2,\dots,p$ is explicitly given by 
$$
\widetilde{\beta}_j = \widehat{\beta}_j - \widehat{\mbu}_j^\top \set{ \frac{1}{n}\sum_{i=1}^n \mbX_i g(\mbX_i^\top \widehat{\bolbeta}_n) - \frac{1}{k}\sum_{i=1}^n \mbX_i \mI(Y_i>Y_{n-k,n})}. 
$$

To establish the asymptotic normality of $\widetilde{\beta}_j$, we impose the following conditions.

 \setcounter{assumption}{0}   
\renewcommand{\theassumption}{(C\arabic{assumption})}
 \begin{assumption}\label{condition:choice:weight} 
As $n\to\infty$, $\lambda_{1,n} \asymp \sqrt{ \log p/k }$ and $\lambda_{2,n} \asymp \sqrt{\log n}$. 
 \end{assumption}

\begin{assumption}\label{high:condition:moment:second}
There exist constants $R>0, \delta>0, \varepsilon>0$ such that
	$$
 \bE \sup_{|t| \le \varepsilon} \abs{g^{\prime\prime}\suit{ \mbX^\top \bolbeta^0+t}}^{1+\delta } \le R, \quad \bE ( g(\mbX^\top \bolbeta^0) )^{1+\delta}\le R.
$$
Moreover, there exists a constant $C$ such that $g(x)\ge C g^\prime(x)$. 
\end{assumption}

\begin{assumption}\label{condition:bias:gradient}
 The function $u\mapsto \Pr(U\le u|\mbX=\mbx)$ is differentiable on $(0,1)$. Denote the derivative by $\phi(u|\mbx)$. Assume that there exist an increasing function $A_2: \mathbb{R}\mapsto \mathbb{R}$ and a function $B_2:\mathbb{R}^p\mapsto \mathbb{R}_{+}$ such that 
$$
\lim_{h\downarrow 0}\frac{1}{A_2(h)} \sup_{\mbx\in \mathcal{X} }\frac{1}{B_2(\mbx)} |\phi(h|\mbx) - \phi(0|\mbx) | =0.  
$$
Moreover, $ A_2(k/n) \log n\to 0 $ as $n\to\infty$ and $\bE |B_2(\mbX)|<\infty$.  
\end{assumption}
 
In addition to the moment conditions around $\boldsymbol{\beta}^0$ in Assumption \ref{high:condition:moment}, we impose a new set of moment conditions around $\boldsymbol{\beta}^0$ in Assumption \ref{high:condition:moment:second}. The second statement of Assumption \ref{high:condition:moment:second} is a regularity condition on the link function, and it is satisfied by links such as $g(x) = \exp(x)$ and $g(x) = \log (1+\exp(x))$.
 Assumption \ref{condition:bias:gradient} regulates the smoothness of the conditional density function near the tail boundary.

\begin{theorem}\label{theorem:normality}
Suppose that the assumptions in Theorem \ref{theorem:high:main} hold. Moreover suppose that Assumption \ref{condition:choice:weight}, \ref{high:condition:moment:second} and \ref{condition:bias:gradient}  hold and $\|\bolbeta^0\|_0\le s$. If $\sqrt{k}/\suit{\log  p \log n  } \to \infty$, 
$\sqrt{k}A(F_Y^{-1}(1-k/n))\sqrt{\log n} =o(1)$,  and  $s=o\!\left\{\sqrt{k}/(\log p\sqrt{\log n})\right\}$, then, as $n\to\infty$,
$$
\sqrt{k}\widehat{v}_j^{-1}\suit{\widetilde{\beta}_j - \beta_j^0 }\stackrel{d}{\to} N(0,1),
$$
where $	\widehat{v}_j^2 = \widehat{\mbu}_j^\top \boldsymbol{\Sigma}_n(\widehat{\bolbeta}_n) \widehat{\mbu}_j $. 
\end{theorem}

Based on Theorem \ref{theorem:normality}, we construct the $1-\alpha$ confidence interval of $\beta_j^0$ as 
$$
[\widetilde{\beta}_j - k^{-1/2}z_{1-\alpha/2} \sqrt{\widehat{\mbu}_j^\top \boldsymbol{\Sigma}_n( \widehat{\bolbeta}_n ) \widehat{\mbu}_j}, \widetilde{\beta}_j + k^{-1/2}z_{1-\alpha/2} \sqrt{\widehat{\mbu}_j^\top \boldsymbol{\Sigma}_n( \widehat{\bolbeta}_n ) \widehat{\mbu}_j}],  
$$
where $z_{1-\alpha/2}$ is the $1-\alpha/2$ quantile of the standard normal distribution.

\section{Simulation}\label{sec:simulation}

\subsection{Setup}\label{sec:simi:setup}

In this section, we conduct numerical simulations to evaluate the finite-sample performance of the proposed estimation and inference procedures. The covariates $(X_2, \dots, X_p)^\top$ are generated from a multivariate normal distribution $N(\boldsymbol{0}, \boldsymbol{\Sigma})$, where the covariance matrix features an autoregressive structure with $\boldsymbol{\Sigma}_{ij} = 0.6^{|i-j|}$ for $i, j = 1, \dots, p-1$. The response variable $Y$ is then generated according to the following mechanism:
$$Y = g(\boldsymbol{X}^\top \boldsymbol{\beta}^0 )Y_0,$$
where the true parameter vector is specified as $$\boldsymbol{\beta}^0 = (\beta_1^0, \beta_2^0, \dots, \beta_p^0 )^\top = (\beta_1^0, 1, 0.5, -1, -0.5, 0, \dots, 0)^\top.$$ 
Here, $\beta_1^0$ is chosen to make $\bE g(\boldsymbol{X}^\top \boldsymbol{\beta}^0 ) = 1$. 
  The term $Y_0$ denotes a heavy-tailed random variable that is independent of $\boldsymbol{X}$.

 For a comprehensive evaluation, we consider three distributions for $Y_0$: the Pareto(1), Fr\'echet(1), and absolute Student-$t$(1) distributions, crossed with two specifications for the link function $g(\cdot)$: the exponential function $g(x)=\exp(x)$, and the Softplus function $g(x) = \log(1+\exp(x))$. All simulation results are based on 5000 Monte Carlo repetitions. These three distributions have tail index one; consequently, multiplication by $g(\mbX^\top\bolbeta^0)$ yields the tail ratio in \eqref{eq:model}. Section \ref{sec:addition:simulation} of the Supplementary Material adds a finite-endpoint uniform design to demonstrate that the method is not restricted to heavy tails.

\subsection{Performance of the $\ell_1$-penalized estimator}

In this subsection, we evaluate the finite-sample performance of the $\ell_1$-penalized estimator $\widehat{\bolbeta}_n$. Under the data-generating process described in Section \ref{sec:simi:setup}, we generate datasets with sample size $n = 2500$ and dimensions $p \in \set{50, 100}$. For the estimation procedure, we set the intermediate sequence to $k = 50$ and consider the tuning parameter $\lambda_n = c_0\sqrt{\log p/k}$ with $c_0 \in \{0.5, 1.0\}$.

To assess the performance, we employ two metrics computed over the covariate coefficients $j \in \{2, \dots, p\}$. The first is the variable selection accuracy, which measures the proportion of correctly identified active and inactive variables:
$$\text{Accuracy} = \frac{1}{p-1} \sum_{j=2}^p \left[ \mI(\widehat{\beta}_{nj} \neq 0) \mI(\beta_j^0 \neq 0) + \mI(\widehat{\beta}_{nj} = 0) \mI(\beta_j^0 = 0) \right].
$$
  The second set of metrics consists of the empirical bias and standard deviation. Specifically, we report these estimation properties for two representative coefficients: $\beta_2^0 = 1$ (an active signal) and $\beta_{10}^0 = 0$ (a noise variable), to reflect the performance on relevant and irrelevant covariates, respectively.

As an idealized baseline for comparison, we also implement the \textit{Oracle estimator} $\widehat{\boldsymbol{\beta}}_n^{\text{or}}$, which assumes the true active support set $S_0 = \{1,2,3,4,5\}$ is known beforehand. This estimator is defined by solving the unpenalized loss function on $S_0$:
$$\widehat{\boldsymbol{\beta}}_{n, S_0}^{\text{or}} = \argmin_{\boldsymbol{\beta}_{S_0}} \left\{ \frac{1}{n}\sum_{i=1}^n G\left(\boldsymbol{X}_{i, S_0}^\top \boldsymbol{\beta}_{S_0}\right) - \frac{1}{k}\sum_{i=1}^n \boldsymbol{X}_{i, S_0}^\top \boldsymbol{\beta}_{S_0} I(Y_i > Y_{n-k,n}) \right\},$$
and setting $\widehat{\beta}_{nj}^{\text{or}} = 0$ for $j \notin S_0$. Because it bypasses high-dimensional noise, its estimation error is invariant to both $p$ and the choice of $\lambda_n$. The finite-sample performance of the oracle estimator is shown in Table \ref{table:oracle}.

\begin{table}
\caption{Simulation results of the oracle estimator}
\label{table:oracle}
\centering
\begin{tabular}[t]{llrr}
\toprule
g & dis & Bias ($\beta_2$) & SD ($\beta_2$)\\
\midrule
exp & Pareto & -0.040 & 0.154\\
exp & Student-t & -0.063 & 0.155\\
exp & Fr\'echet & -0.078 & 0.155\\
softplus & Pareto & -0.014 & 0.318\\
softplus & Student-t & -0.010 & 0.321\\
softplus & Fr\'echet & -0.020 & 0.316\\
\bottomrule
\end{tabular}
\end{table}

The simulation results for the $\ell_1$-penalized estimator are summarized in Tables \ref{tab:simulation:exp} and \ref{tab:simulation:log}. First, compared with the oracle estimator, the $\ell_1$-penalized estimator has larger bias but smaller standard deviation. This is expected because the $\ell_1$-penalty introduces a shrinkage effect that reduces variance at the expense of bias. Second, the scaling factor $c_0$ controls the trade-off between variable selection and parameter estimation. Specifically, choosing $c_0 = 1.0$ yields higher variable selection accuracy, exceeding 95\% in most configurations. However, this larger penalty also leads to a larger estimation bias for the active coefficient $\beta_2^0$.

Third, the estimator shows stable performance across different dimensions $p$ and error distributions. Notably, the bias under the Pareto distribution is slightly smaller than that under the Fr\'echet and Student-$t$ distributions. This difference arises because the Pareto error distribution has no approximation error in model \eqref{eq:model}. By contrast, under Fr\'echet and Student-$t$ distributions, the estimator has two sources of bias: the shrinkage bias from the $\ell_1$-penalty and the model approximation bias.

In conclusion, the proposed $\ell_1$-penalized estimator balances sparsity recovery and estimation accuracy. We recommend $c_0 = 1.0$ if the primary goal is variable selection, and $c_0 = 0.5$ if minimizing estimation error is preferred.

 \begin{table}[htbp]
  \centering
  \caption{Simulation results of the $\ell_1$-penalized estimator under $g(x)=\exp(x)$.}
  \label{tab:simulation:exp}
  \footnotesize
  \begin{tabular}{lrrrrrrr}
\toprule
dis & p & $c_0$ & accuracy & bias($\beta_2$) & sd ($\beta_2$) & bias($\beta_{10}$) & sd ($\beta_{10}$)\\
\midrule
Pareto & 50 & 0.5 & 0.798 & -0.149 & 0.146 & -0.001 & 0.040\\
Pareto & 50 & 1.0 & 0.953 & -0.296 & 0.138 & -0.001 & 0.012\\
Pareto & 100 & 0.5 & 0.824 & -0.161 & 0.146 & -0.001 & 0.036\\
Pareto & 100 & 1.0 & 0.970 & -0.322 & 0.138 & 0.000 & 0.008\\
Student-t & 50 & 0.5 & 0.797 & -0.176 & 0.149 & -0.001 & 0.040\\
Student-t & 50 & 1.0 & 0.952 & -0.325 & 0.140 & -0.001 & 0.011\\
Student-t & 100 & 0.5 & 0.820 & -0.189 & 0.143 & 0.000 & 0.034\\
Student-t & 100 & 1.0 & 0.969 & -0.358 & 0.138 & 0.000 & 0.011\\
Fr\'echet & 50 & 0.5 & 0.795 & -0.189 & 0.149 & 0.000 & 0.041\\
Fr\'echet & 50 & 1.0 & 0.951 & -0.342 & 0.136 & 0.000 & 0.011\\
Fr\'echet & 100 & 0.5 & 0.819 & -0.205 & 0.147 & -0.002 & 0.037\\
Fr\'echet & 100 & 1.0 & 0.969 & -0.374 & 0.137 & 0.000 & 0.009\\
\bottomrule
\end{tabular}
\end{table}

\begin{table}[htbp]
  \centering
  \caption{Simulation results of the $\ell_1$-penalized estimator under $g(x)=\log(1+\exp(x))$.}
  \label{tab:simulation:log}
  \footnotesize
  \begin{tabular}{lrrrrrrr}
\toprule
dis & p & $c_0$ & accuracy & bias($\beta_2$) & sd ($\beta_2$) & bias($\beta_{10}$) & sd ($\beta_{10}$)\\\midrule
Pareto & 50 & 0.5 & 0.768 & -0.324 & 0.262 & -0.001 & 0.072\\
Pareto & 50 & 1.0 & 0.929 & -0.666 & 0.208 & -0.001 & 0.019\\
Pareto & 100 & 0.5 & 0.796 & -0.344 & 0.259 & -0.002 & 0.071\\
Pareto & 100 & 1.0 & 0.957 & -0.717 & 0.197 & 0.000 & 0.013\\
Student-t & 50 & 0.5 & 0.768 & -0.325 & 0.262 & -0.003 & 0.077\\
Student-t & 50 & 1.0 & 0.930 & -0.667 & 0.209 & -0.001 & 0.022\\
Student-t & 100 & 0.5 & 0.796 & -0.341 & 0.258 & -0.003 & 0.067\\
Student-t & 100 & 1.0 & 0.956 & -0.716 & 0.196 & -0.001 & 0.015\\
Fr\'echet & 50 & 0.5 & 0.768 & -0.343 & 0.259 & -0.002 & 0.077\\
Fr\'echet & 50 & 1.0 & 0.930 & -0.680 & 0.201 & -0.002 & 0.020\\
Fr\'echet & 100 & 0.5 & 0.795 & -0.357 & 0.260 & -0.003 & 0.066\\
Fr\'echet & 100 & 1.0 & 0.956 & -0.726 & 0.193 & -0.001 & 0.017\\
\bottomrule
\end{tabular}
\end{table}

 \subsection{Performance of the debiased estimator}\label{sec:simu:ci}

In this subsection, we evaluate the finite-sample performance of the proposed debiased estimator $\widetilde{\beta}_j$ and its associated confidence intervals. Under the data-generating process described in Section \ref{sec:simi:setup}, we generate datasets with a total sample size of $2n$, where $n=2500$, and consider dimensions $p \in \{50, 100\}$. The intermediate tail sequence is fixed at $k = 50$. For the high-dimensional regularized estimation step, the tuning parameter is specified as $\lambda_n = c_0\sqrt{\log p/k}$. Since our primary objective is statistical inference rather than variable selection, we employ a relatively smaller regularization strength by setting $c_0=0.5$.

To select $\lambda_{1,n}=c_1\sqrt{\log p/k}$, we arrange the candidate values of $c_1$ in decreasing order and sequentially tighten the constraint. Starting from the largest candidate, we continue to the next smaller value whenever the optimization problem remains feasible and stop when feasibility is first lost. The chosen $c_1$ is the smallest feasible value. In addition, set $\lambda_{2n}=10\sqrt{\log n}$. Here the constant $10$ is relatively large because the second constraint is a technical safeguard and should not affect the variance minimization.

To empirically verify the asymptotic normality established in Theorem \ref{theorem:normality}, we examine the quantile-to-quantile (QQ) plots for the studentized statistics defined as
\begin{equation*}
\frac{\widetilde{\beta}_j - \beta_j^0}{\sqrt{\widehat{\mbu}_j^\top \boldsymbol{\Sigma}_n( \widehat{\boldsymbol{\beta}}_n ) \widehat{\mbu}_j/k} }.
\end{equation*}
The plots are displayed in Figures \ref{fig:all_qqplots:exp} and \ref{fig:all_qqplots:softplus}. For all configurations, the empirical quantiles exhibit a tight alignment with the $y=x$ reference line, suggesting that the standard normal distribution serves as an excellent approximation for the finite-sample distribution of the proposed debiased estimator.

\begin{figure}[htbp]
    \centering
    \begin{subfigure}{\textwidth}
        \centering
        \includegraphics[width=0.9\textwidth]{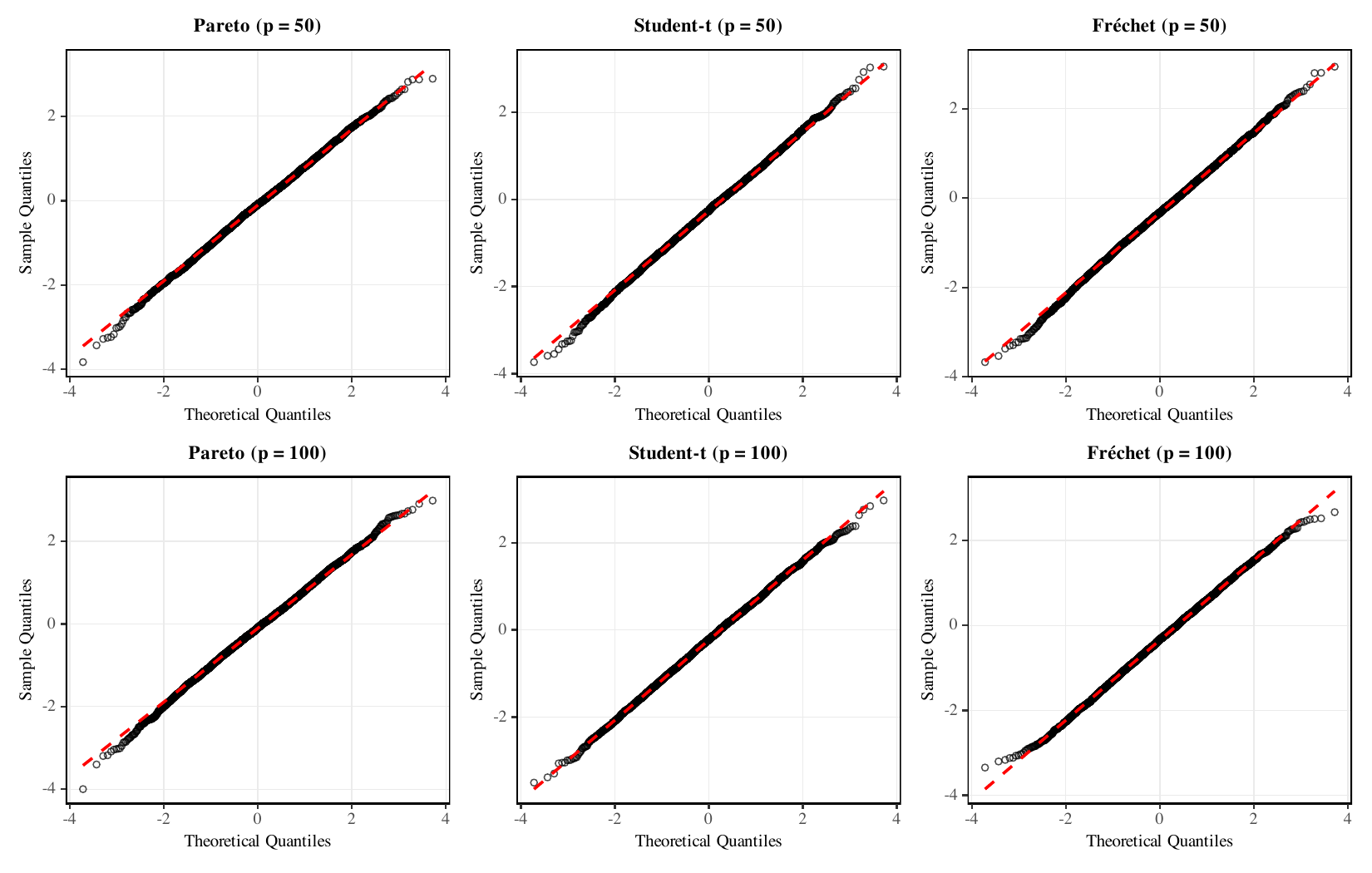}
        \caption{QQ plots for coefficient $\beta_2$.}
    \end{subfigure}    
    \vspace{1em} 
    \begin{subfigure}{\textwidth}
        \centering
        \includegraphics[width=0.9\textwidth]{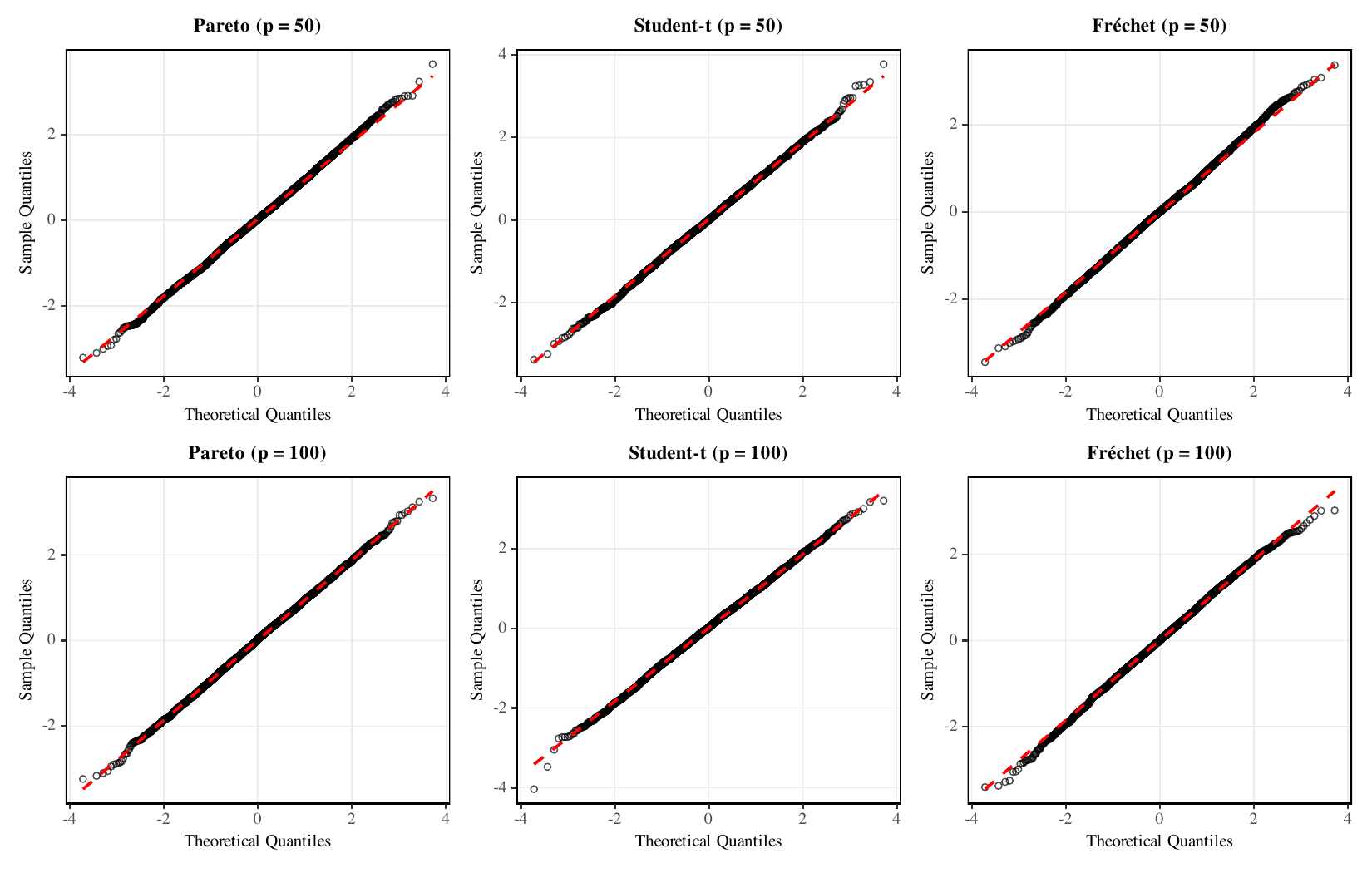}
        \caption{QQ plots for coefficient $\beta_{10}$.}
    \end{subfigure}
    \caption{QQ plots for the debiased estimator for the case $g(x)=\exp(x)$.}
    \label{fig:all_qqplots:exp}
\end{figure}

 \begin{figure}[htbp]
    \centering
    \begin{subfigure}{\textwidth}
        \centering
        \includegraphics[width=0.9\textwidth]{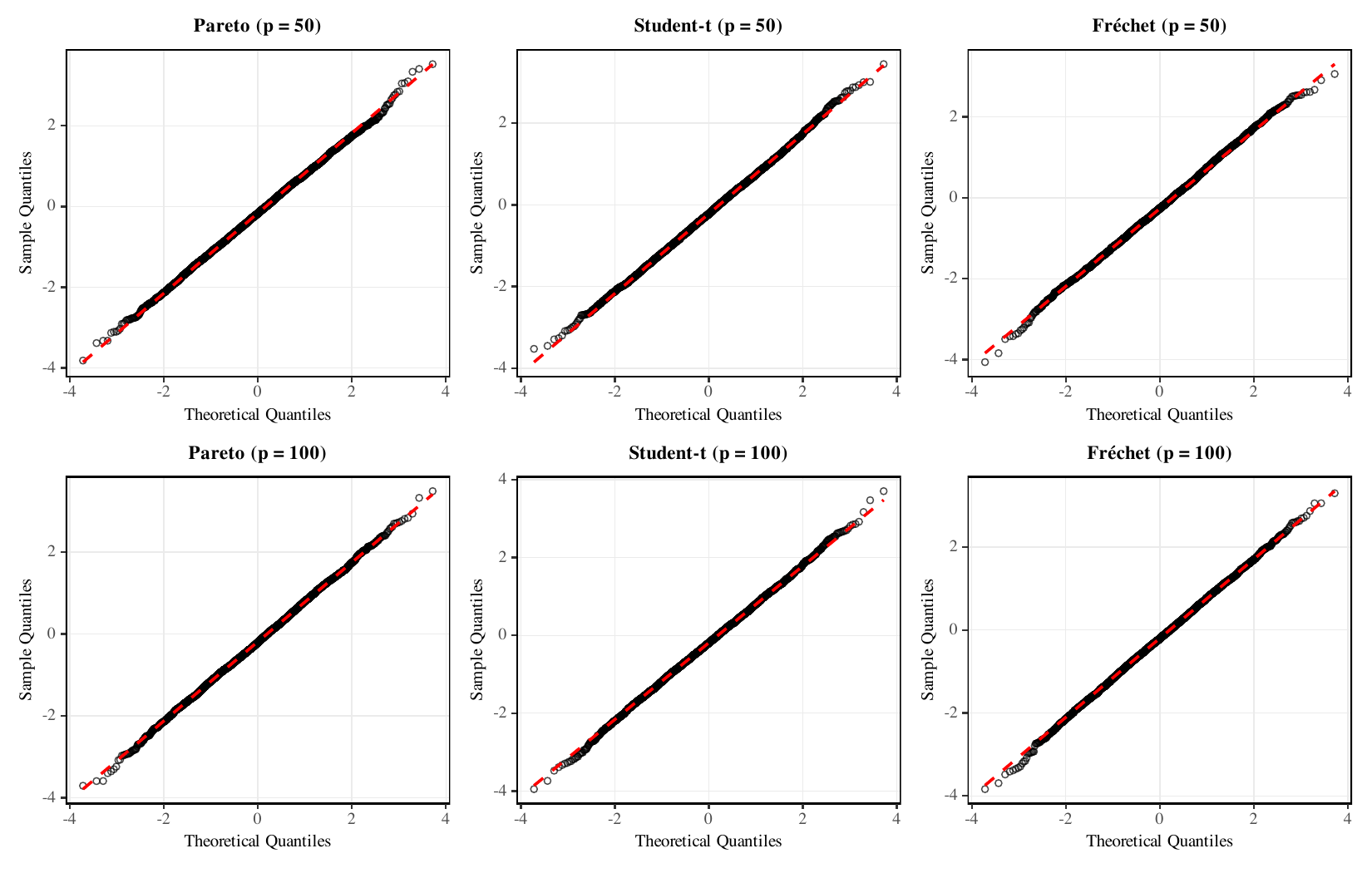}
        \caption{QQ plots for coefficient $\beta_2$.}
    \end{subfigure}    
    \vspace{1em} 
    \begin{subfigure}{\textwidth}
        \centering
        \includegraphics[width=0.9\textwidth]{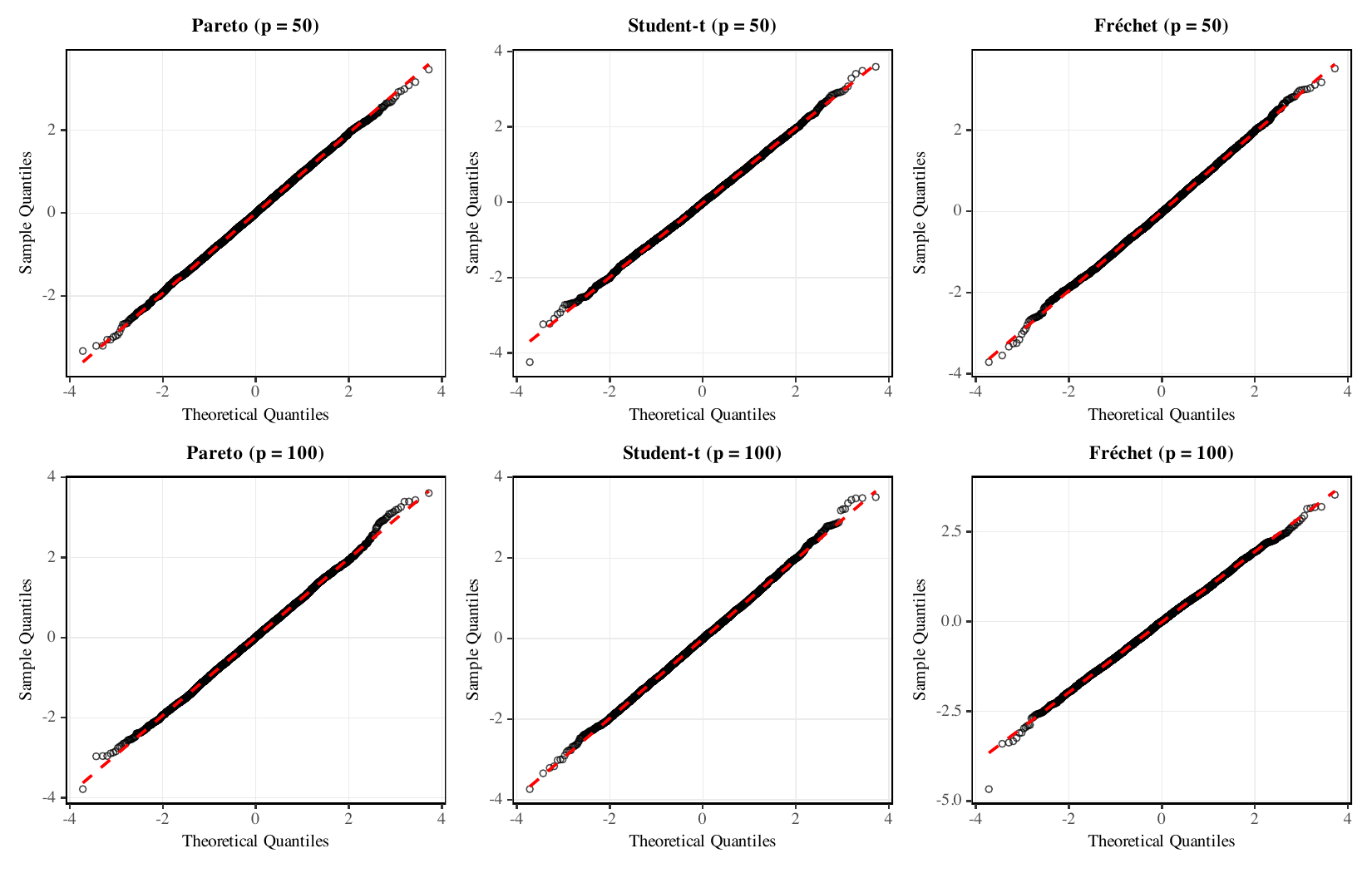}
        \caption{QQ plots for coefficient $\beta_{10}$.}
    \end{subfigure}
    \caption{QQ plots for the debiased estimator for the case $g(x)=\log(1+\exp(x))$.}
    \label{fig:all_qqplots:softplus}
\end{figure}

In addition, we evaluate the finite-sample inference performance by reporting the empirical biases, standard deviations (SDs), and coverage probabilities of the nominal $95\%$ confidence intervals. The results are summarized in Table \ref{tab:inference_results}. Several observations can be drawn. 

First, the debiasing procedure substantially reduces the shrinkage bias of the $\ell_1$-penalized estimator and produces coverage close to the nominal level across the reported designs. Second, the active coefficient $\beta_2$ generally has a smaller absolute bias under Pareto errors than under Student-$t$ and Fr\'echet errors. This difference arises because the Pareto error distribution has no approximation error in model \eqref{eq:model}, whereas the other two distributions also introduce model approximation bias. Third, for the noise coordinate $\beta_{10}=0$, empirical biases are close to zero and coverage is near the nominal level. These results support the finite-sample calibration of the procedure in the simulated settings.

 \begin{table}[htbp]
  \centering
\caption{Simulation results of the debiased estimator.}
\label{tab:inference_results}
    \begin{subtable}{\textwidth} \centering
    \caption{$g(x)=\exp(x)$}  
\begin{tabular}{lccccccc}
\toprule
dis & p & coverage($\beta_2$) & bias($\beta_2$) & sd($\beta_2$) & coverage($\beta_{10}$) & bias($\beta_{10}$) & sd($\beta_{10}$)\\
\midrule
Pareto & 50 & 0.965 & -0.018 & 0.155 & 0.963 & 0.005 & 0.183\\
Pareto & 100 & 0.964 & -0.018 & 0.157 & 0.964 & 0.002 & 0.185\\
Student-t & 50 & 0.958 & -0.044 & 0.157 & 0.958 & 0.002 & 0.185\\
Student-t & 100 & 0.961 & -0.039 & 0.155 & 0.963 & 0.003 & 0.183\\
Fr\'echet & 50 & 0.955 & -0.055 & 0.154 & 0.959 & 0.000 & 0.186\\
Fr\'echet & 100 & 0.949 & -0.058 & 0.160 & 0.960 & 0.001 & 0.186\\
\bottomrule
\end{tabular}
  \end{subtable}
  
  \vspace{0.2cm}  
  \begin{subtable}{\textwidth}    
    \centering
    \caption{ $g(x)=\log(1+\exp(x))$} 
  \begin{tabular}{lccccccc}
\toprule
dis & p & coverage($\beta_2$) & bias($\beta_2$) & sd($\beta_2$) & coverage($\beta_{10}$) & bias($\beta_{10}$) & sd($\beta_{10}$)\\
\midrule
Pareto & 50 & 0.954 & -0.051 & 0.296 & 0.958 & 0.001 & 0.329\\
Pareto & 100 & 0.950 & -0.058 & 0.299 & 0.958 & 0.003 & 0.340\\
Student-t & 50 & 0.950 & -0.063 & 0.296 & 0.954 & -0.004 & 0.336\\
Student-t & 100 & 0.946 & -0.058 & 0.307 & 0.952 & -0.001 & 0.341\\
Fr\'echet & 50 & 0.949 & -0.075 & 0.297 & 0.958 & -0.003 & 0.334\\
\addlinespace
Fr\'echet & 100 & 0.953 & -0.061 & 0.296 & 0.956 & -0.005 & 0.336\\
\bottomrule
\end{tabular}
  \end{subtable}

\end{table}

\subsection{Using the Hessian as the variance matrix}\label{sec:simu:falsify}

A common practice in generalized linear models is to invoke the information matrix identity and use the Hessian as the variance of the score. This identity does not hold in the present tail-localized setting; see Remark \ref{remark:exp:rankone}. To illustrate the resulting effect on inference, we consider a naive Hessian-based procedure. Specifically, we replace the estimated score covariance $\boldsymbol{\Sigma}_n(\widehat{\boldsymbol{\beta}}_n)$ by the Hessian $\mathcal{L}_n^{\prime\prime}(\widehat{\boldsymbol{\beta}}_n)$ throughout the construction of the projection direction. We retain the same tuning-parameter grids and feasibility constraints as in the proposed procedure. Thus, the naive projection direction is selected by minimizing
\begin{equation*}
\widehat{\boldsymbol{u}}_{j,H}^{\top}
\mathcal{L}_n^{\prime\prime}(\widehat{\boldsymbol{\beta}}_n)
\widehat{\boldsymbol{u}}_{j,H}
\end{equation*}
over the same feasible set used for $\widehat{\boldsymbol{u}}_j$. The corresponding debiased estimator is constructed using $\widehat{\boldsymbol{u}}_{j,H}$, and its standard error is calculated as
\begin{equation*}
\widehat{\mathrm{se}}_{j,H}
=
\sqrt{
\widehat{\boldsymbol{u}}_{j,H}^{\top}
\mathcal{L}_n^{\prime\prime}(\widehat{\boldsymbol{\beta}}_n)
\widehat{\boldsymbol{u}}_{j,H}/k
}.
\end{equation*}
We report the empirical coverage probabilities of the resulting nominal $95\%$ confidence intervals and compare them with those of the proposed variance-based procedure. This comparison isolates the consequence of incorrectly treating the Hessian as the score variance. We focus on the softplus link as explained in Remark \ref{remark:exp:rankone}.

\begin{table}[htbp]
\centering
\caption{Simulation results of the Hessian-based procedure under $g(x)=\log(1+\exp(x))$.}
\label{tab:inference:hessian}
\begin{tabular}{lccccccc}
\toprule
dis & p & coverage($\beta_2$) & bias($\beta_2$) & sd($\beta_2$) & coverage($\beta_{10}$) & bias($\beta_{10}$) & sd($\beta_{10}$)\\
\midrule
Pareto & 50  & 0.866 & -0.053 & 0.296 & 0.876 &  0.001 & 0.329\\
Pareto & 100 & 0.858 & -0.062 & 0.300 & 0.858 &  0.003 & 0.339\\
Student-t & 50  & 0.863 & -0.065 & 0.297 & 0.871 & -0.004 & 0.336\\
Student-t & 100 & 0.850 & -0.063 & 0.308 & 0.865 & -0.001 & 0.341\\
Fr\'echet & 50  & 0.854 & -0.077 & 0.297 & 0.869 & -0.003 & 0.334\\
Fr\'echet & 100 & 0.859 & -0.065 & 0.296 & 0.868 & -0.005 & 0.336\\
\bottomrule
\end{tabular}
\end{table}

Table~\ref{tab:inference:hessian} shows substantial undercoverage when the Hessian is used as the variance matrix, with coverage well below the nominal level across the reported configurations. By comparison, the proposed variance-based procedure produces coverage close to the nominal level under the same configurations.

\section{Real Data Analysis}\label{sec:realdata}

In this section, we apply the proposed heteroscedastic extremes framework to an auto insurance claims dataset to investigate how various high-dimensional risk factors drive the tail scaling behaviors of claim sizes. The high heterogeneity of insurance data typically arises from the diverse nature of underlying risks and multiple factors influencing claim frequencies and sizes. Although this heterogeneity complicates the modeling process, it also provides critical insights into risk profiles. 

The dataset, available on Kaggle\footnote{https://www.kaggle.com/datasets/xiaomengsun/car-insurance-claim-data}, consists of $n=8423$ observations. The response variable $Y$ represents the auto claim amount. Following the preprocessing steps in \cite{tang2024high}, the original design matrix contained 44 covariates, including an intercept term ($X_1 \equiv 1$). A detailed description of these variables is provided in Section S.4 of the Supplementary Material of \cite{tang2024high}. Because the claims-history indicators were perfectly linearly dependent, we removed \texttt{old.claims.1}, which indicates that the total claim amount over the previous five years lies in $(0,5000)$, leaving $p=43$ covariates in the analysis.  Except for the intercept $X_1$, all covariates are standardized to have a mean of zero and a variance of one.

We consider two different link functions, $g(x) = \exp(x)$ and $g(x) = \log(1+\exp(x))$, for the heteroscedastic extremes models. First, we apply the $\ell_1$-penalized estimator $\widehat{\bolbeta}_n$ to evaluate the predictive performance of our model. We set $k=50$ and $\lambda_n = 0.5\sqrt{\log p/k}$. Since $p=43$ is comparable to the effective sample size $k=50$, this application represents a high-dimensional setting.

 We randomly divide the dataset into a training set ($50\%$ of the data, $n_1 = 4211$ observations) and a holdout set (the remaining $50\%$, $n_2 = 4212$ observations). The current implementation evaluates the following cross-entropy criterion
$$
\text{PE} = -\frac{1}{n_2}\sum_{j\in \mathcal{J}}\suit{\mI\suit{Y_j> y_{th} }\log \widehat{p}_j + \mI\suit{Y_j\le y_{th} } \log (1-\widehat{p}_j) },
$$
where $\mathcal{J}$ is the index set of the holdout set, $y_{th}$ is the $98\%$ quantile of $Y$ in the holdout set, and the estimated probability of $Y_j>y_{th}$ is defined as
$$
\widehat{p}_j = g\suit{\mbX_j^\top \widehat{\bolbeta}_n } \suit{\frac{1}{n_2}\sum_{i \in \mathcal{J}} \mI\suit{Y_i>y_{th}}}.
$$

We compare our two proposed models with a baseline model without covariates, where 
$$
\widehat{p}_j = \frac{1}{n_2}\sum_{i \in \mathcal{J}} \mI\suit{Y_i>y_{th}}.
$$
Additionally, we compare against an $\ell_1$-penalized logistic regression model, with the penalty parameter selected via cross-validation.

The random-splitting evaluation is repeated 500 times, and the resulting PEs are reported in Table \ref{tab:l1:car}. The exponential link yields a slightly smaller prediction error than the Softplus link, and both proposed models outperform the no-covariate and penalized-logistic benchmarks.

\begin{table}[htbp]
\centering
\caption{Mean prediction errors (with standard deviations in parentheses; $\times 100$) for different methods based on 500 random splits of the auto claims data.}
\label{tab:l1:car}
\begin{tabular}{cccccc}
\hline
Methods & $g(x) = \exp(x)$ & $g(x)=\log(1+\exp(x))$   & No Covariates & Penalized Logistic \\ \hline
PE      & 9.35 (0.135)     & 9.38 (0.148)           &       9.87 (0.004) & 9.80 (0.123)  \\ \hline
\end{tabular}
\end{table}

We then repeat the sample-splitting procedure 50 times and apply the proposed debiased inference method to each split. The tuning parameters are chosen in the same manner as in Section \ref{sec:simulation}.
For each model, a variable is identified as significant if the corresponding confidence interval excludes zero in at least 30 of the 50 splits (60\%), with the direction of the effect remaining the same. Both models identify the same set of significant risk factors, as shown in Table \ref{tab:debiased:ci}.

To account for the additional uncertainty induced by random sample splitting, we follow the median aggregation method of \cite{chernozhukov2018double}. For each coefficient, let $\widehat\beta_j^{(b)}$ and $\widehat{\mathrm{se}}_j^{(b)}$ denote the debiased estimate and its standard error from split $b$, respectively, and define $\widetilde\beta_j=\operatorname{median}_{1\leq b\leq 50}\widehat\beta_j^{(b)}$. The split-adjusted variance estimator is
\[
\widetilde V_j
=
\operatorname{median}_{1\leq b\leq 50}
\left\{
\bigl(\widehat{\mathrm{se}}_j^{(b)}\bigr)^2
+
\bigl(\widehat\beta_j^{(b)}-\widetilde\beta_j\bigr)^2
\right\}.
\]
Table~\ref{tab:debiased:ci} reports $\widetilde\beta_j$ together with the corresponding 95\% confidence interval $\widetilde\beta_j\pm1.96\sqrt{\widetilde V_j}$ under each model. All of the adjusted confidence intervals for the reported variables exclude zero, confirming their statistical significance after accounting for the uncertainty induced by sample splitting.


%

Identifying these significant risk factors can help insurers monitor conditional tail risks and improve risk-management strategies. In particular, the results distinguish higher- and lower-risk policies according to individual and vehicle characteristics and can provide useful information for insurance pricing and the management of large-loss exposure.

\begin{table}[htbp]
\centering
\setlength{\tabcolsep}{4pt}
\renewcommand{\arraystretch}{0.9}
\caption{Significant variables, median debiased estimates and sample-splitting-adjusted 95\% confidence intervals under the two models, and variable descriptions. The intervals use the median variance correction of \cite{chernozhukov2018double}. Variables are selected when they are significant in at least 60\% of the 50 splits.}
\label{tab:debiased:ci}
\fontsize{10pt}{10.5pt}\selectfont
\begin{tabularx}{\textwidth}{
    l
    c
    c
    >{\leavevmode\fontsize{10pt}{10.5pt}\selectfont}X
}
\hline
Variables
& \begin{tabular}[c]{@{}c@{}}$g(x)=\exp(x)$\\ Median (95\% CI)\end{tabular}
& \begin{tabular}[c]{@{}c@{}}$g(x)=\log(1+\exp(x))$\\ Median (95\% CI)\end{tabular}
& Description \\
\hline

\multicolumn{4}{l}{\textbf{Positive Effects}} \\[2pt]

log.Vehicle.Value
& 0.180 [0.069, 0.291] & 0.343 [0.126, 0.560] & log(Vehicle Value) \\

log.Vehicle.Value2
& 0.161 [0.054, 0.268] & 0.314 [0.102, 0.526] & $\log^2$(Vehicle Value) \\

\hline
\multicolumn{4}{l}{\textbf{Negative Effects}} \\
\noalign{\vskip 5pt}

City.Population
& -0.532 [-0.904, -0.160] & -0.873 [-1.460, -0.285] & equals 1 if living in city areas \\

is.Manager
& -0.426 [-0.779, -0.073] & -0.746 [-1.333, -0.160] & equals 1 if the occupation is manager \\ 
is.minivan
& -0.424 [-0.809, -0.039] & -0.750 [-1.423, -0.077] & equals 1 if the car is a minivan \\
\hline
\end{tabularx}
\end{table}

\section{Discussion}\label{sec:discussion}

Model \eqref{eq:model} assumes that the conditional distributions of $Y$ 
are driven by the covariates. The estimation theory in Section 
\ref{sec:estimation} controls the tail approximation error through the 
second-order condition in Assumption \ref{all:model:bias}. There is no 
formal test for this model. A specification test can be based on testing 
whether conditional tail indices are constant across the covariate 
space, similar to the test in \cite{einmahl2016statistics}. A rigorous testing 
procedure for the model is left for future research.

A natural next step following this study is inference on conditional 
exceedance probabilities and extreme conditional quantiles beyond the 
observed range. Within the present framework this would combine 
Weissman-type extrapolation of the marginal tail (see, e.g., 
\cite{haan2006extreme}) with the estimated function 
$c_{\bolbeta}(\mbx)$. Theoretical justification will follow, by using 
the delta method. However, applying the delta method is not routine 
here: the extrapolation factor depends on the whole vector $\bolbeta$ 
which involves the full score covariance $\bolSigma_n$ in Section 
\ref{sec:inference}, not merely the single-coefficient variance of 
Theorem \ref{theorem:normality}. Remark \ref{remark:exp:rankone} shows 
how the intercept already shapes this covariance under the canonical 
exponential link. We leave this extension to future work.

\bibliographystyle{apalike} 
\bibliography{mybib.bib}

\appendix
\clearpage
\renewcommand{\theequation}{S.\arabic{equation}}
\setcounter{equation}{0}

\renewcommand{\thefigure}{S.\arabic{figure}}
\setcounter{figure}{0}

\renewcommand{\thetable}{S.\arabic{table}}
\setcounter{table}{0}

\setcounter{page}{1}
\renewcommand{\thepage}{S\arabic{page}}
\renewcommand{\thelemma}{S\arabic{lemma}}

Throughout the proofs, we use the notation $y_n = F_Y^{-1}(1-k/n)$. For two sequences $a_n$ and $b_n$, $a_n\lesssim b_n$ means that $a_n\le c b_n$ for all sufficiently large $n$ and some constant $c>0$.

\section{Proof of Theorem \ref{fix:theorem:normality}}

\begin{proof}[Proof of Theorem \ref{fix:theorem:normality}]
	
Let $\Delta \in \mathbb{R}^p$ and	define 
\begin{align*}
	V_n( \Delta) = & k\left\{\mL_n\suit{\bolbeta^0+ k^{-1/2}\Delta  }  - \mL_n(\bolbeta^0) + \lambda_n \left \|\bolbeta^0+  k^{-1/2}\Delta  \right \|_1 -\lambda_n \| \bolbeta^0\|_1\right\}.
\end{align*}
Then, 
$$
\sqrt{k}\suit{\widehat{\bolbeta}_n - \bolbeta^0 } = \argmin_{\Delta} V_n(\Delta).
$$

We intend to establish the asymptotic behavior of $V_n(\Delta)$ for $\|\Delta\|_2 \le C$, where $C>0$ is a fixed constant. Define $U_i = 1-F_Y(Y_i)$. Then, $U_i$ follows a uniform distribution on the interval $[0,1]$, and $Y_i>Y_{n-k,n}$ is equivalent to $ U_i\le U_{k,n}$. 
Define
$$
G_n(\bolbeta) = \frac{1}{n}\sum_{i=1}^n G\suit{\mbX_i^\top \bolbeta}.
$$ 
By applying the Taylor expansion to the function $G_n$, we have that
\begin{align*}
	&V_n(\Delta) \\
	=& k\suit{ G_n(\bolbeta^0+k^{-1/2}\Delta)- G_n(\bolbeta^0) } -\frac{1}{\sqrt{k}}\sum_{i=1}^n  \Delta^\top \mbX_i\mI(Y_i>Y_{n-k,n} ) +  k\lambda_n \sum_{j=1}^p \suit{|\beta_j^0+k^{-1/2}\Delta_j| -|\beta_j^0| }   \\
=& \sqrt{k} \Delta^\top \nabla G_n(\bolbeta^0) +  \frac{1}{2}\Delta^\top \nabla^2 G_n( \widetilde{\bolbeta} )\Delta -\frac{1}{\sqrt{k}}\sum_{i=1}^n  \Delta^\top \mbX_i\mI(U_i\le U_{k,n} ) + k\lambda_n \sum_{j=1}^p \suit{|\beta_j^0+k^{-1/2}\Delta_j| -|\beta_j^0| }  \\
=&:J_1(\Delta)+J_2(\Delta)+J_3(\Delta)+J_4(\Delta)-J_5(\Delta),
\end{align*}
 where $\widetilde{\bolbeta} = \bolbeta^0 + tk^{-1/2}\Delta$ for some $t\in [0,1]$, and
\begin{align*}
	J_1(\Delta)=&  \sqrt{k} \Delta^\top\set{  \nabla G_n(\bolbeta^0) -  \bE \suit{\mbX g(\mbX^\top\bolbeta^0) }  }, \\
	J_2(\Delta)=& \frac{1}{2} \Delta^\top \nabla^2 G_n( \widetilde{\bolbeta} )\Delta, \\
		J_3(\Delta)=&  k\lambda_n \sum_{j=1}^p \suit{|\beta_j^0+k^{-1/2} \Delta_j| -|\beta_j^0| },  \\
		J_4(\Delta)=&  \sqrt{k}\Delta^\top \set{ \bE \suit{\mbX g(\mbX^\top\bolbeta^0) } - \bE(\mbX|U\le k/n) }, \\
	J_5(\Delta)=& \sqrt{k}\Delta^\top  \suit{\frac{1}{k}\sum_{i=1}^n  \mbX_i\mI(U_i\le U_{k,n} ) -\bE(\mbX|U\le k/n) }.
\end{align*}

The proof strategy for bounding $J_1(\Delta)$, $J_2(\Delta)$, and $J_3(\Delta)$ is conceptually aligned with the analysis of Lasso-type estimators in generalized linear models \citep{knight2000asymptotics, zou2006adaptive}. In contrast, the treatments of $J_4(\Delta)$ and $J_5(\Delta)$ rely on extreme value techniques.

First, we handle $J_1(\Delta)$. Note that
 \begin{align*}
	J_1(\Delta) = &\sqrt{k}\Delta^\top \suit{   \nabla G_n(\bolbeta^0) - \bE \suit{\mbX g(\mbX^\top\bolbeta^0) } } \\
	=& \sqrt{k} \Delta^\top\set{ \frac{1}{n} \sum_{i=1}^n \mbX_i g\suit{\mbX_i^\top \bolbeta^0 }     -  \bE \suit{\mbX g(\mbX^\top\bolbeta^0) }}.
\end{align*} 
 By Assumption \ref{fix:condition:moment:bound}, we have that 
 \begin{align*}
 n\text{Var} \suit{\frac{1}{n} \sum_{i=1}^n X_{i,j} g\suit{ \mbX_i^\top  \bolbeta^0 }}  =& \text{Var}\suit{X_j g\suit{\mbX^\top  \bolbeta^0}  }\\
 \le & \bE \suit{ X_j^2 g^2\suit{\mbX^\top  \bolbeta^0} }<\infty.
 \end{align*}
 Thus, we have that as $n\to\infty$,
 $$
 \sup_{\|\Delta\|_2\le C} |J_1(\Delta)| = o_P(1).
 $$

 Next, we handle $J_2(\Delta)$. Note that
 \begin{align*}
 	\nabla^2 G_n(\widetilde{\bolbeta}) = & \frac{1}{n} \sum_{i=1}^n  \mbX_i\mbX_i^\top g^\prime\suit{ \mbX_i^\top \widetilde{\bolbeta}}.
 \end{align*}
 As $n\to\infty$, we have that 
 $\|\widetilde{\bolbeta} -\bolbeta^0\|_2 = k^{-1/2} t\|\Delta\|_2 \to 0$.
Then, by Assumption \ref{fix:condition:moment:bound} and the law of large numbers for a triangular array, we have that 
$$
\nabla^2 G_n(\widetilde{\bolbeta}) = \boldsymbol{H}(\bolbeta^0) +o_P(1).
$$  
Thus, we conclude that, as $n\to\infty$, 
$$
J_2(\Delta) = \frac{1}{2} \Delta^\top \boldsymbol{H}(\bolbeta^0) \Delta +o_P(1).
$$

 For $J_3(\Delta)$, we have that
   $$
   \begin{aligned}
   	J_3(\Delta) =& k\lambda_n \sum_{j=1}^p \suit{|\beta_j^0+ k^{-1/2}\Delta_j| -|\beta_j^0| } \\ \to & \lambda_0\sum_{j=1}^p \suit{\Delta_j\text{sgn}(\beta_j^0)\mI\suit{\beta_j^0\ne 0}+|\Delta_j|\mI(\beta_j^0=0) }\\
   	 =: &b(\Delta), 
   \end{aligned}
$$
uniformly for all $\|\Delta\|_2\le C$.

 Next, we handle $J_4(\Delta)$. By Bayes' rule,     
\begin{equation}\label{s:eq:bayes}
\begin{aligned}
	\frac{1-F_Y(y|\mbX=\mbx)}{ 1-F_Y(y)} = \frac{f_{\mbX}(\mbx|Y>y)}{f_{\mbX}(\mbx)}.
\end{aligned}
\end{equation}
By Assumption \ref{all:model:bias} and $\sqrt{k}A(y_n)\to 0$ , we have that as $n\to\infty$,

\begin{equation*} 
	\begin{aligned}
	\bE(\mbX|U\le k/n) 
	=& \int \mbx f_{\mbX}(\mbx| U \le  k/n)d\mbx \\
	=& \int \mbx f_{\mbX}(\mbx)  g(\mbx^\top\bolbeta^0 )  d\mbx + O(1)A(y_n ) \int \mbx f_{\mbX}(\mbx) B(\mbx)d\mbx \\
 =&\bE \suit{\mbX g(\mbX^\top\bolbeta^0 )} + o(1)k^{-1/2}.
	\end{aligned}
\end{equation*}
Thus, we conclude that, as $n\to\infty$, 
$$
 \sup_{\|\Delta\|_2\le C} |J_4(\Delta)| = o(1).
$$

 Finally, we handle $J_5(\Delta)$.
 Under Assumptions \ref{all:model:bias} and \ref{fix:condition:empiricalprocess}, the conditions of Theorem 3 in \cite{aghbalou2024tail} are satisfied. Therefore, by Theorem 3 of \cite{aghbalou2024tail}, we have that as $n\to\infty$, 
 $$
 \sqrt{k} \suit{\frac{1}{k}\sum_{i=1}^n \mbX_i\mI(U_i\le U_{k,n} ) -\bE(\mbX|U\le k/n) } = \boldsymbol{W}+o_p(1), 
 $$
 where $\boldsymbol{W}$ is a mean-zero Gaussian random vector with covariance 
 $$
 \begin{aligned}
 	 \boldsymbol{\Sigma}_{\boldsymbol{W}}=& \lim_{y\to y^{+}} \set{\bE(\mbX\mbX^\top | Y\ge y) - \bE(\mbX|Y\ge y)  \bE(\mbX^\top|Y\ge y)}, \\
 	 =& \bE \suit{ \mbX\mbX^\top  g(\mbX^\top \bolbeta^0)  } - \bE \suit{ \mbX g(\mbX^\top \bolbeta^0)  } \bE \suit{ \mbX^\top  g(\mbX^\top \bolbeta^0)  }.
 \end{aligned}
 $$
Therefore, we have that 
  $$
  J_5(\Delta) = \Delta^\top \boldsymbol{W} +o_P(1),
  $$
  uniformly for all $\|\Delta\|_2 \le C$.

Combining the limits of $J_1(\Delta), J_2(\Delta),J_3(\Delta), J_4(\Delta)$ and $J_5(\Delta)$, we conclude that, as $n\to\infty$, 
$$
V_n(\Delta) \stackrel{d}{\to} V(\Delta) = -\Delta^\top \boldsymbol{W} +\frac{1}{2} \Delta^\top \boldsymbol{H}(\bolbeta^0) \Delta +b(\Delta).
$$
 Since $V_n$ is convex and $V$ has a unique minimum, by Theorem 5 of \cite{knight1999epi}, we have that as $n\to\infty$,  
 $$
 \sqrt{k}\suit{\widehat{\bolbeta}_n - \bolbeta^0 } = \argmin_{\Delta} V_n(\Delta) \stackrel{d}{\to} \argmin_{\Delta}V(\Delta). 
 $$
 The proof is then complete. 
\end{proof}

\section{Proofs of Theorem \ref{theorem:high:main} }

\subsection{Preliminary Lemmas}

The following lemma gathers some well-known properties of sub-Gaussian random variables. The proofs are omitted. For more details, see, for example, \cite{Wain2019high}. 
\begin{lemma}\label{lemma:sub-gaussian}
	Assume that $X$ is a mean-zero sub-Gaussian random variable with parameter $K$. Then, 
\begin{itemize}
	\item[(i)] For any positive integer  $b\ge 1$,
	$$
		\bE |X|^b \le  2K^{b}\Gamma(b/2+1),
	$$ 
	where $\Gamma$ denotes the gamma function.
	\item[(ii)] For any $t>0$, 
	$$
		\Pr(|X|>t) \le 2\exp\suit{ -\frac{t^2}{K^2} }.
	$$
\end{itemize}
\end{lemma}

\begin{lemma}\label{lemma:bernstein}(Bernstein inequality, see e.g., \cite{shorack1986empirical}, page 855)
Let $X_1,\dots, X_n$ be independent zero-mean random variables. Suppose that $|X_i| \le M$ almost surely, for all $i=1,\dots,n$. Then, for all positive $t$,
$$
\Pr\suit{\sum_{i=1}^n X_i \ge t} \le \exp\suit{-\frac{\frac{1}{2}t^2}{ \sum_{i=1}^n \bE X_i^2+\frac{1}{3}Mt } }.
$$
\end{lemma}
 \begin{lemma}\label{lemma:hessian:population}
	For any $\bolbeta, \Delta$ satisfying $\|\bolbeta\|_2\le R$, $\| \Delta\|_2\le 1$, there exists a constant $C>0$ such that
	 $$
	 \Delta^\top \bE \suit{\mbX \mbX^\top g^\prime\suit{\mbX^\top \bolbeta }} \Delta \ge C \|\Delta\|_2^2.  
	 $$
\end{lemma}
\begin{proof}[Proof of Lemma \ref{lemma:hessian:population}] 
Let $T>0$ be a constant to be specified later. Then 
\begin{align*}
		 \Delta^\top  \bE \suit{\mbX \mbX^\top  g^\prime\suit{\mbX^\top \bolbeta }} \Delta=& \bE \suit{  g^\prime\suit{\mbX^\top \bolbeta } \suit{\mbX^\top \Delta }^2 } \\
	 \ge & \bE \suit{  g^\prime\suit{\mbX^\top \bolbeta } \suit{\mbX^\top \Delta }^2 \mI\suit{ |\mbX^\top\bolbeta| \le T } } \\
	 \ge & g^\prime\suit{-T} \bE \suit{\suit{\mbX^\top \Delta }^2 \mI\suit{ |\mbX^\top\bolbeta|\le T} } \\
	 =&  g^\prime\suit{-T}\set{  \bE \suit{\mbX^\top \Delta }^2 - \bE \suit{\suit{\mbX^\top \Delta }^2 \mI\suit{ |\mbX^\top\bolbeta|> T} }  }.
\end{align*}
By Assumption \ref{condition:regular:x}(i), we have that 
$$
 \bE \suit{\mbX^\top \Delta }^2 \ge \kappa_{l} \|\Delta\|_2^2. 
$$
By the Cauchy--Schwarz inequality, we have that
\begin{align*}
	\bE \suit{\suit{\mbX^\top \Delta }^2 \mI\suit{ |\mbX^\top\bolbeta|> T} } \le & \sqrt{ \bE \suit{\mbX^\top \Delta }^4 \Pr\suit{ |\mbX^\top\bolbeta|> T } }.
\end{align*}
By Assumption \ref{condition:regular:x}(iii) and Lemma \ref{lemma:sub-gaussian}, we have that 
\begin{align*}
	 \bE \suit{\mbX^\top \Delta }^4   \le &4 \kappa_u^4 \|\Delta\|_2^4,\\
	 \Pr\suit{|\mbX^\top \bolbeta|>T } \le & 2\exp\suit{ - \frac{T^2}{2\kappa_u^2 \|\bolbeta\|_2^2 }} \le 2\exp\suit{ - \frac{T^2}{2\kappa_u^2 R^2 }}.
\end{align*}
It follows that, 
$$
	\bE \suit{\suit{\mbX^\top \Delta }^2 \mI\suit{ |\mbX^\top\bolbeta|> T} } \le 2\sqrt{2} \|\Delta\|_2^2 \kappa_u^2\exp\suit{ -\frac{T^2}{4 \kappa_u^2 R^2} }.
$$
By choosing 
$$T^2 = 4\kappa_u^2 R^2\log \suit{\frac{8\kappa_u^2}{\kappa_l} },
$$ we have that 
$$
\bE \suit{\suit{\mbX^\top \Delta }^2 \mI\suit{ |\mbX^\top\bolbeta|> T} } \le \frac{\kappa_{\ell}}{2} \|\Delta\|_2^2.
$$
The proof is then complete. 

\end{proof}

\begin{lemma}\label{lemma:bound:young}
Let $X$ be a mean-zero sub-Gaussian random variable with parameter $K$. For any event $A$ with $\Pr(A) =\alpha$, it holds that, 
$$
|\bE (X\mI_A)|\le 2K\alpha \sqrt{\log (2/\alpha)}.
$$  
\end{lemma}
\begin{proof}[Proof of Lemma \ref{lemma:bound:young}]
	For any $t_0 \ge 0$, by Lemma \ref{lemma:sub-gaussian}, we have
	$$
	\begin{aligned}
	|\bE(X \mI_A)| \le & \int_0^\infty \Pr(A \cap \{|X| > t\})  dt  \\
	\le &  \int_0^{t_0} \Pr(A) \, dt + \int_{t_0}^\infty \Pr(|X| > t)  dt \\
	\le & \alpha t_0 + 2\int_{t_0}^\infty \exp\left(-\frac{t^2}{K^2}\right)  dt. 		
	\end{aligned}
	$$
Note that
$$
\begin{aligned}
	\int_{t_0}^\infty \exp\left(-\frac{t^2}{K^2}\right)  dt \le &   \int_{t_0}^\infty \frac{t}{t_0} \exp\left(-\frac{t^2}{K^2}\right)dt \\
	 =& \frac{1}{t_0} \left[-\frac{K^2}{2}\exp\suit{-\frac{t^2}{K^2}}  \right]_{t_0}^\infty \\
	 =& \frac{K^2}{2t_0}\exp\suit{-\frac{t_0^2}{K^2}}.
\end{aligned}
$$

Setting $t_0 = K \sqrt{\log(2/\alpha)} \ge 0$ yields

$$
\begin{aligned}
	|\bE\suit{X \mI_A}| \le & K \alpha \sqrt{\log(2/\alpha)} +   \frac{K\alpha}{2\sqrt{\log (2/\alpha)} } \\
	=& K \alpha \sqrt{\log(2/\alpha)} \suit{1+\frac{1}{2\log(2/\alpha)} } \\
	\le &  2K \alpha \sqrt{\log(2/\alpha)},
	\end{aligned}
$$
The proof is then complete.
\end{proof}

\subsection{Proof of Theorem \ref{theorem:high:main}}
\begin{proof}[Proof of Theorem \ref{theorem:high:main}]
	We apply Theorem 1 of \cite{negahban2012unified} to establish the bound on $\|\widehat{\bolbeta}_n- \bolbeta^0\|_2$. To this end, we first verify the restricted strong convexity condition for our loss function $\mL_n$: there exist constants $c_1>0,c_2>0,c_3>0$ such that with probability at least $1-c_1\exp(-c_2 n)$,
	\begin{equation}\label{s:rsc}
	\delta \mL_n(\Delta) \ge 	c_3 \|\Delta\|_2^2, \quad \text{for all } \Delta\in \mathbb{C}, 
	\end{equation}
	where
	$$
\delta \mL_n (\Delta):= \mL_n(\bolbeta^0+\Delta) - \mL_n(\bolbeta^0) - \Delta^\top \nabla \mL_n(\bolbeta^0), 
	$$
	and 
	$$
	\mathbb{C}= \set{\Delta: \|\Delta\|_2\le 1, \|\Delta_{S^c}\|_1 \le 3 \|\Delta_S\|_1 }.
	$$ 	
	Here, $S$ is the support set of $\bolbeta^0$, i.e., $S = \set{j: \beta_j^0\ne 0 }$, and $S^c = \set{j: \beta_j^0= 0 }$.

Note that 
\begin{align*}
	\delta \mL_n (\Delta)= & \mL_n(\bolbeta^0+\Delta) - \mL_n(\bolbeta^0) - \Delta^\top \nabla \mL_n(\bolbeta^0) \\
	=& \frac{1}{n}\sum_{i=1}^n G\suit{\mbX_i^\top \suit{\bolbeta^0+ \Delta} } - \frac{1}{k}\sum_{i=1}^n \mbX_i^\top \suit{ \bolbeta^0 +\Delta}\mI\suit{Y_i>Y_{n-k,n} } \\
	 &- \frac{1}{n}\sum_{i=1}^n G\suit{\mbX_i^\top \bolbeta^0 } + \frac{1}{k}\sum_{i=1}^n \mbX_i^\top \bolbeta^0\mI\suit{Y_i>Y_{n-k,n} } \\
	 &- \frac{1}{n}\sum_{i=1}^n \mbX_i^\top \Delta g\suit{ \mbX_i^\top \bolbeta^0} + \frac{1}{k}\sum_{i=1}^n \mbX_i^\top \Delta \mI\suit{Y_i>Y_{n-k,n} } \\
	 =& \frac{1}{n}\sum_{i=1}^n G\suit{\mbX_i^\top \suit{\bolbeta^0+ \Delta} } - \frac{1}{n}\sum_{i=1}^n G\suit{\mbX_i^\top \bolbeta^0 } - \frac{1}{n}\sum_{i=1}^n \mbX_i^\top \Delta g\suit{ \mbX_i^\top \bolbeta^0} \\
	 =& G_n(\bolbeta^0+\Delta) - G_n(\bolbeta^0) -  \Delta^\top \nabla G_n(\bolbeta^0).
\end{align*}	
Applying a Taylor expansion to the function	$G_n$ at point $\bolbeta^0$, we have that 
$$
\delta \mL_n (\Delta) = \frac{1}{2n} \sum_{i=1}^n g^\prime\suit{ \mbX_i^\top \bolbeta^0 +t \mbX_i^\top \Delta } \suit{\mbX_i^\top \Delta}^2,
$$
for some constant $t\in [0,1]$. By equation (43) of \cite{negahban2012unified} (see Proposition 2 of \cite{negahban2013unifiedframeworkhighdimensionalanalysis} for details), there exist constants $\kappa_1>0, \kappa_2>0$ such that
$$
\delta \mL_n (\Delta) \ge \kappa_1 \|\Delta\|_2^2 -\kappa_2 \frac{\log p}{n} \|\Delta\|_1^2, \quad \text{for all} \ \|\Delta\|_2\le 1.
$$
On the cone $\mathbb{C}$, we have that 
$$
\|\Delta\|_1 \le 4 \|\Delta_{S}\|_1\le 4\sqrt{s} \|\Delta_S\|_2\le 4 \sqrt{s}\|\Delta\|_2. 
$$	
It follows that,
$$
\delta \mL_n (\Delta) \ge \suit{\kappa_1-16\kappa_2 \frac{s\log p}{n}} \|\Delta\|_2^2, \quad \text{for all} \ \Delta \in \mathbb{C}.
$$
Since $s\log p/n = o(1)$, we conclude that \eqref{s:rsc} holds.

By Lemma \ref{lemma:gradient:infinity}, we have that with probability greater than $1-n^{-c_1} -p^{-c_2}$,  
 $$
  \| \nabla \mL_n(\bolbeta^0)\|_{\infty} \lesssim \lambda_n.
  $$
Then, by Theorem 1 of \cite{negahban2012unified}, we have that with probability tending to 1,
$$
\|\widehat{\bolbeta}_n - \bolbeta^0\|_2 \lesssim \sqrt{s\lambda_n^2} \lesssim \sqrt{\frac{s\log p}{k}}.
$$ 
 The proof is then complete.

\end{proof}

\begin{lemma}\label{lemma:gradient:infinity}
	Assume the same conditions as in Theorem \ref{theorem:high:main}. Then, there exist constants $c_1>0, c_2>0, c_3>0$ such that for sufficiently large $n$, 
\begin{equation*}
	\Pr\suit{ \| \nabla \mL_n(\bolbeta^0)\|_{\infty} \ge c_1 \sqrt{\log p/k} }\le n^{-c_2} + p^{-c_3}.
\end{equation*} 
\end{lemma}
\begin{proof}[Proof of Lemma \ref{lemma:gradient:infinity}]
Throughout the proof, the constants $c_1,c_2, c_3$ may differ from line to line.
By the triangle inequality, we have that  
\begin{align*}
	 \| \nabla \mL_n(\bolbeta^0)\|_{\infty}  = & \|\frac{1}{n}\sum_{i=1}^n  \mbX_i  g\suit{\mbX_i^\top \bolbeta^0}   - \frac{1}{k}\sum_{i=1}^n  \mbX_i \mI\suit{Y_i>Y_{n-k,n} }      \|_{\infty} \\
	 \le &  \|\frac{1}{n}\sum_{i=1}^n  \set{\mbX_i  g\suit{\mbX_i^\top \bolbeta^0}-\bE\suit{\mbX g\suit{\mbX^\top \bolbeta^0} }   }  \|_{\infty} \\
	 &+ \| \bE\suit{\mbX g\suit{\mbX^\top \bolbeta^0} }  -\bE \suit{\mbX|Y\ge F_Y^{-1}(1-k/n)}   \|_{\infty} \\
	 &+ \| \frac{1}{k} \sum_{i=1}^n \mbX_i \mI\suit{Y_i>Y_{n-k,n} }  - \frac{1}{k} \sum_{i=1}^n \mbX_i \mI\suit{Y_i\ge y_n}       \|_{\infty}\\
	 &+\| \frac{1}{k} \sum_{i=1}^n \mbX_i \mI\suit{Y_i\ge y_n} -  \bE (\mbX|Y\ge y_n)     \|_{\infty} \\
	 =&: J_1+J_2+ J_3+ J_4.
\end{align*}
  We complete the proof by showing that,   
\begin{align*}
	&\Pr\suit{J_1 > c_1 \sqrt{\log p/k}} \le     n^{-c_2}+  p^{-c_3}, \\
	&J_2 \le  c_1\sqrt{\log p/k},\\
	&\Pr\suit{J_3 > c_1\sqrt{\log p/k}} \le       p^{-c_3}, \\
	&\Pr\suit{J_4 > c_1\sqrt{\log p/k}} \le     p^{-c_3}.
\end{align*}

First, we handle $J_1$.  
Denote 
$$
V_{ij} = X_{ij}g\suit{\mbX_i^\top \bolbeta^0} - \bE X_{ij} g\suit{\mbX_i^\top \bolbeta^0}
$$
Define the event
$$
\mathcal{E}_T = \set{ \max_{1\le i\le n} | \mbX_i^\top \bolbeta^0| \le T\sqrt{\log n} },
$$
where $T$ is constant to be specified later. 
By Lemma \ref{lemma:sub-gaussian}, we have that 
\begin{align*}
	1- \Pr\suit{\mathcal{E}_T  }  \le  & n\max_{1\le i\le n} \Pr\suit{  | \mbX_i^\top \bolbeta^0| \ge T\sqrt{\log n}}\\
	\le &  n\exp \suit{ -\frac{T^2 \log n}{\kappa_u^2 \|\bolbeta^0\|_2^2 } } 
	\le \exp \suit{ -\frac{T^2 \log n}{\kappa_u^2 R^2 }+\log n }.
\end{align*}
Choose $T$ such that $T^2> \kappa_u^2 R^2$. Then, for some $c_2>0$, 
$$
 \Pr\suit{\mathcal{E}_T } \ge 1-n^{-c_2}.
$$

On the event $\mathcal{E}_T$, by Assumption \ref{condition:regular:x}(ii) and Assumption \ref{high:condition:moment}, we have that 
$$
V_{ij} \le R\suit{ g\suit{ T\sqrt{\log n} }+ |\bE g\suit{\mbX^\top\bolbeta^0 }| } \le C g\suit{ T\sqrt{\log n} } =:M_n.
$$
Then, by the Bernstein inequality (Lemma \ref{lemma:bernstein}), we have that on the event $\mathcal{E}_T$,
$$
\begin{aligned}
	\Pr\suit{ \max_{1\le j\le p}|\frac{1}{n}\sum_{i=1}^n V_{ij}|> t } \le 2p\max_{1\le j\le p} \exp\suit{-\frac{\frac{1}{2} nt^2}{ \bE V_{ij}^2 +\frac{1}{3} M_n t } }\le 2 \exp\suit{-\frac{\frac{1}{2} nt^2}{ C +\frac{1}{3} M_n t } +\log p }.
\end{aligned}
$$
Take $t = c_1 \sqrt{\log p/n}$. Then, 
$$
-\frac{\frac{1}{2} nt^2}{ C +\frac{1}{3} M_n t } +\log p = -\frac{\frac{1}{2}c_1^2 \log p}{C+ \frac{1}{3}M_n c_1\sqrt{\log p/n}} +\log p.
$$
By Assumption \ref{high:condition:moment}, we have that as $n\to\infty$,
$$
M_n\sqrt{\log p/n} \to 0.  
$$
Thus, by choosing $c_1$ such that $\frac{1}{2}c_1^2>C$, we obtain, for some $c_3>0$, 
$$
\exp\suit{-\frac{\frac{1}{2} nt^2}{ C +\frac{1}{3} M_n t } +\log p } \le \exp(- c_3\log p) = p^{-c_3}. 
$$

For $J_{2}$, by Assumption \ref{all:model:bias} and \eqref{s:eq:bayes}, 
we have that 
$$
\begin{aligned}
J_2=	&\| \bE\suit{\mbX g\suit{\mbX^\top \bolbeta^0} }  -\bE (\mbX|Y\ge F_Y^{-1}(1-k/n)) \|_{\infty} \\
 =& 	\| \bE\suit{\mbX g\suit{\mbX^\top \bolbeta^0} }  -\int \mbx f_{\mbX}(\mbx|Y\ge F_Y^{-1}(1-k/n))d\mbx  \|_{\infty} \\
 =&\| \bE\suit{\mbX g\suit{\mbX^\top \bolbeta^0} }  -\int \mbx f_{\mbX}(\mbx)\frac{1-F_Y(y_n|\mbX=\mbx)}{1-F_Y(y_n)} d\mbx \|_{\infty} \\
 =& O(1) A(y_n) \| \int \mbx B(\mbx)   f_{\mbX}(\mbx) d\mbx  \|_{\infty} \\
 =&  o(1) \sqrt{\log p/k}.
 \end{aligned}
$$

Next, we handle $J_3$. Note that
\begin{align*}
	\abs{\frac{1}{k}\sum_{i=1}^n X_{ij} \suit{ \mI\suit{Y_i > Y_{n-k,n} }-  \mI\suit{Y_i \ge  y_n}}} \le &C\frac{1}{k} \sum_{i=1}^n \abs{ \mI\suit{Y_i \ge y_n}-\mI\suit{Y_i > Y_{n-k,n} }} \\
	=& C\abs{ \frac{1}{k} \sum_{i=1}^n \mI\suit{Y_i \ge y_n}-\frac{1}{k}\sum_{i=1}^n \mI\suit{Y_i > Y_{n-k,n} }   } \\
	=& C\abs{ \frac{1}{k} \sum_{i=1}^n \mI\suit{Y_i \ge y_n}- 1  }.
\end{align*}
 The first equality holds because the differences $\mI\suit{Y_i > Y_{n-k,n} }- \mI\suit{Y_i \ge y_n}$, $i=1,\dots,n$, are all of the same sign or equal to zero and therefore do not change sign across $i$.

 By the Bernstein inequality (Lemma \ref{lemma:bernstein}), we have that 
 $$
 \Pr\suit{\abs{ \frac{1}{k} \sum_{i=1}^n \mI\suit{Y_i \ge y_n}- 1 } > t} \le 2\exp\suit{ -\frac{k^2t^2}{ k+ \frac{1}{3}kt} }.
 $$
 Taking $t = c_1\sqrt{\log p/k}$ gives
$$
\frac{k^2t^2}{ k+ \frac{1}{3}kt} \ge c_3 \log p.
$$
Thus, 
$$
\Pr(J_3>c_1\sqrt{\log p/k}) \le p^{-c_3}.
$$

Finally, we handle $J_4$. Denote
$$
A_{ij} = X_{ij} \mI\suit{Y_i\ge y_n} -\bE(X_{j}\mI\suit{Y\ge y_n})
$$ 
Note that $A_{ij} \le 2\kappa_X$ and  
$$
\begin{aligned}
	\bE A_{ij}^2 \le  & \bE X_{ij}^2 \mI(Y_i\ge y_n)\le \kappa_X^2k/n.
\end{aligned}
$$
 By the Bernstein inequality (Lemma \ref{lemma:bernstein}), we have that for any $t>0$, 
 \begin{align*}
\Pr\suit{\abs{\frac{1}{k}\sum_{i=1}^n A_{ij}}> t } =&  \Pr\suit{\abs{\sum_{i=1}^n A_{ij}}> kt } \\
\le & 2\exp\suit{ -\frac{\frac{1}{2}k^2t^2 }{ \kappa_X^2k + \frac{2}{3}\kappa_Xkt } } =2\exp\suit{-k \frac{t^2}{2\kappa_X^2 +\frac{4}{3}\kappa_Xt}  }.
\end{align*}
Since this holds for each $j=1,\dots, p$, we have that 
\begin{align*}
	\Pr\suit{ J_4> t } =& \Pr\suit{\max_{1\le j\le p}\abs{\frac{1}{k}\sum_{i=1}^n A_{ij}}> t } \\
	\le & 2p \exp\suit{-k \frac{t^2}{2\kappa_X^2 +\frac{4}{3}\kappa_Xt}  }.
\end{align*}
Taking $t = c_1\sqrt{\log p/k}$ with $c_1>3\kappa_X^2$ gives 
$$
	\Pr\suit{ J_4> c_1\sqrt{\log p/k} } \le p^{-c_3}. 
$$
The proof is then complete.

\end{proof}

\section{Proof of Theorem \ref{theorem:normality}}

First, we show that the optimization problem is non-empty with probability tending to 1. 
\begin{lemma}\label{lemma:exist}
Assume the same conditions as in Theorem \ref{theorem:normality}. Then, with probability tending to 1, we have the following results.  
\begin{itemize}
	\item[(i)] The matrix $ \boldsymbol{\Gamma}_n$ is invertible, where $\boldsymbol{\Gamma}_n = \bE \suit{\mbX \mbX^\top g^\prime(\mbX^\top \widehat{\bolbeta}_n )|\widehat{\bolbeta}_n }.$
	\item[(ii)] Let $u_j^*$ denote the $j$-th column of $ \boldsymbol{\Gamma}_n^{-1}$. The vector $u_j^*$ satisfies the constraints \eqref{bias:basis} and \eqref{bias:max}. 
\end{itemize}
\end{lemma}
\begin{proof}[Proof of Lemma \ref{lemma:exist}]

By Theorem \ref{theorem:high:main}, we have that with probability tending to 1, $\| \widehat{\bolbeta}_n\|_2 \le R$, for some $R>0$. Thus, by Lemma \ref{lemma:hessian:population}, we have that  
  for some 
$\delta>0$, 
\begin{equation*}
		\omega^\top \boldsymbol{\Gamma}_n \omega = \bE \suit{ g^\prime(\mbX^\top \widehat{\bolbeta}_n ) (\mbX^\top \omega)^2|\widehat{\bolbeta}_n } >\delta, 
	\end{equation*}
	 for any unit vector $\omega \in \mathbb{R}^p$. Thus, (i) holds. 
 
 For (ii), note that, conditional on $\widehat{\bolbeta}_n$, the random variables $g^\prime(\mbX_i^\top \widehat{\bolbeta}_n )(u_j^*)^\top \mbX_i\mbX_i^\top e_k $, $i=1,\dots,n$ are independent with mean $\delta_{jk}=1$ if $j=k$ and $\delta_{jk}=0$ otherwise. Similar to the handling of $J_1$ in the proof of Lemma \ref{lemma:gradient:infinity}, we can show that, $u_j^*$ satisfies the constraint \eqref{bias:basis}. 
 By (i), we have that $\|\boldsymbol{\Gamma}_n^{-1}\|_2 = O(1)$ and hence $\|u_j^*\|_2=O(1)$. Then, by the sub-Gaussianity of $\mbX$ (Assumption \ref{condition:regular:x} and Lemma \ref{lemma:sub-gaussian}), we have that 
  $u_j^*$ satisfies the constraint \eqref{bias:max}.
 The proof is then complete. 

\end{proof}

\begin{proof}[Proof of Theorem \ref{theorem:normality}]
	Write 
\begin{align*}
		\widetilde{\beta}_j - \beta_j^0  
	=&	\widehat{\beta}_j-\beta_j^0 - \widehat{\mbu}_j^\top \set{ \frac{1}{n}\sum_{i=1}^n \mbX_i g(\mbX_i^\top \widehat{\bolbeta}_n) - \frac{1}{k}\sum_{i=1}^n \mbX_i \mI(Y_i>Y_{n-k,n})}\\
	=&  - \widehat{\mbu}_j^\top \set{ \frac{1}{n}\sum_{i=1}^n \mbX_i g(\mbX_i^\top \bolbeta^0) - \frac{1}{k}\sum_{i=1}^n \mbX_i \mI(U_i\le U_{k,n})} \\
	&- \suit{ \frac{\widehat{\mbu}_j^\top}{n}\sum_{i=1}^n \mbX_i\mbX_i^\top g^\prime(\mbX_i^\top \widehat{\bolbeta}_n) -e_j  }\suit{\widehat{\bolbeta}_n - \bolbeta^0 } \\
	&-\widehat{\mbu}_j^\top \set{\frac{1}{n}\sum_{i=1}^n\mbX_i g(\mbX_i^\top \widehat{\bolbeta}_n ) - \frac{1}{n}\sum_{i=1}^n\mbX_i g(\mbX_i^\top \bolbeta^0 )-  \frac{1}{n}\sum_{i=1}^n \mbX_i\mbX_i^\top g^\prime(\mbX_i^\top \widehat{\bolbeta}_n)\suit{\widehat{\bolbeta}_n - \bolbeta^0 } } \\
	=&: I_1 +I_2 +I_3.
\end{align*}

We intend to show that, as $n\to\infty$, $I_2 = o_P(1/\sqrt{k})$ and $I_3=o_P(1/\sqrt{k})$. We start with $I_2$. By \eqref{bias:basis} and Assumption \ref{condition:choice:weight}, we have that 
$$
I_2 \lesssim \sqrt{\log p/k} \|\widehat{\bolbeta}_n - \bolbeta^0\|_1.
$$
By Corollary 1 of \cite{negahban2012unified}, whose conditions are verified in the proof of Theorem \ref{theorem:high:main}, we have that with probability tending to 1,  
\begin{equation}\label{s:hat:l1:norm}
	\| \widehat{\bolbeta}_n- \bolbeta^0 \|_1 \lesssim \frac{s\sqrt{\log p}}{\sqrt{ k}}.
\end{equation}
Consequently, as $n\to\infty$, with probability tending to one,
$$
 I_2 \lesssim \frac{ s\log p}{k} = o(k^{-1/2}).
$$

Next, we handle $I_3$. By the mean-value theorem, we have that 
\begin{align*}
	I_3 =& -\widehat{\mbu}_j^\top \set{\frac{1}{n}\sum_{i=1}^n\mbX_i g(\mbX_i^\top \widehat{\bolbeta}_n ) - \frac{1}{n}\sum_{i=1}^n\mbX_i g(\mbX_i^\top \bolbeta^0 )-  \frac{1}{n}\sum_{i=1}^n \mbX_i\mbX_i^\top g^\prime(\mbX_i^\top  \bolbeta^0)\suit{\widehat{\bolbeta}_n - \bolbeta^0 } } \\
	&-\widehat{\mbu}_j^\top  \frac{1}{n}\sum_{i=1}^n \mbX_i\mbX_i^\top \suit{g^\prime(\mbX_i^\top  \widehat{\bolbeta}_n) - g^\prime(\mbX_i^\top  \bolbeta^0)}\suit{\widehat{\bolbeta}_n - \bolbeta^0 } \\
	=& -  \frac{1}{2n}\sum_{i=1}^n \widehat{\mbu}_j^\top \mbX_i g^{\prime\prime}\suit{\mbX_i^\top  \bolbeta^0+t_{1i}\mbX_i^\top  \suit{\widehat{\bolbeta}_n-\bolbeta^0 }} \suit{\mbX_i^\top\suit{\widehat{\bolbeta}_n-\bolbeta^0 }  }^2 \\
	&  -  \frac{1}{n}\sum_{i=1}^n \widehat{\mbu}_j^\top  \mbX_i g^{\prime\prime}\suit{\mbX_i^\top  \bolbeta^0+t_{2i}\mbX_i^\top  \suit{\widehat{\bolbeta}_n-\bolbeta^0 }} \suit{\mbX_i^\top\suit{\widehat{\bolbeta}_n-\bolbeta^0 }  }^2 \\
	=&: I_{31}+I_{32},
\end{align*}
where $t_{1i},t_{2i}\in [0,1]$.

By \eqref{bias:max} and H\"older's inequality, we have that for any $\delta>0$, 
$$
I_{31} = O_P(1) \sqrt{\log n} \suit{\frac{1}{n}\sum_{i=1}^n |\Delta_i|^{1+\delta}}^{1/(1+\delta)} \set{\frac{1}{n}\sum_{i=1}^n \suit{\mbX_i^\top\suit{\widehat{\bolbeta}_n-\bolbeta^0 } }^{2(1+1/\delta) }}^{\delta/(1+\delta)},
$$
where 
$$
\Delta_i = g^{\prime\prime}\suit{\mbX_i^\top \bolbeta^0+t_{1i}\mbX_i^\top \suit{\widehat{\bolbeta}_n-\bolbeta^0 }}. 
$$
By Assumption \ref{condition:regular:x} and Theorem \ref{theorem:high:main}, we have that with probability tending to 1, 
$$
\set{\frac{1}{n}\sum_{i=1}^n \suit{\mbX_i^\top\suit{\widehat{\bolbeta}_n-\bolbeta^0 } }^{2(1+1/\delta) }}^{\delta/(1+\delta)} \lesssim \frac{s\log p}{k}.
$$
By Assumption \ref{condition:regular:x}(ii) and \eqref{s:hat:l1:norm}, we have that 
$$
\abs{t_{1i}\mbX_i^\top \suit{\widehat{\bolbeta}_n-\bolbeta^0 }} \le \kappa_X \| \widehat{\bolbeta}_n-\bolbeta^0\|_1 \lesssim \suit{\frac{s^2\log p}{k}}^{1/2} \to 0.
$$
Then, by Assumption \ref{high:condition:moment:second} and Markov inequality, we have that 
with probability tending to 1, 
$$
\suit{\frac{1}{n}\sum_{i=1}^n |\Delta_i|^{1+\delta}}^{1/(1+\delta)} \lesssim 1. 
$$
Thus, we conclude that, with probability tending to 1, 
$$
I_{31} \lesssim \sqrt{\log n} \frac{s\log p}{k} = o(k^{-1/2}). 
$$
Similarly, we have that with probability tending to 1, 
 $I_{32} =o(k^{-1/2})$. Thus, we conclude that, as $n\to\infty$, with probability tending to 1, 
$$
I_3 = o(k^{-1/2}). 
$$  

Combining the limits of $I_2$ and $I_3$, we have that with probability tending to 1,
$$
\widetilde{\beta}_j - \beta_j^0 = - \widehat{\mbu}_j^\top \set{ \frac{1}{n}\sum_{i=1}^n \mbX_i g(\mbX_i^\top \bolbeta^0) - \frac{1}{k}\sum_{i=1}^n \mbX_i \mI(Y_i>Y_{n-k,n})}+o(k^{-1/2}).
$$
By Lemma \ref{lemma:normality} below, we then complete the proof.

\end{proof}

\begin{lemma}\label{lemma:normality}
Assume the same conditions as in Theorem \ref{theorem:normality}. Then, as $n\to\infty$, 
$$
\frac{\sqrt{k}}{\widehat v_j} \set{\frac{1}{k}\sum_{i=1}^n \widehat{\mbu}_j^\top\mbX_i \mI(U_i\le U_{k,n}) - \frac{1}{n}\sum_{i=1}^n \widehat{\mbu}_j^\top\mbX_i g(\mbX_i^\top \bolbeta^0) } \stackrel{d}{\to} N(0, 1).
$$

\end{lemma}
\begin{proof}[Proof of Lemma \ref{lemma:normality}]
	
Define 
$$
\widetilde{\Gamma}_n(u) = \sqrt{k}\suit{\widetilde{D}_n(u) -D_n(u) },$$
where
$$
  \widetilde{D}_n(u) = \frac{1}{k}\sum_{i=1}^n Z_i\mI\set{U_i\le ku/n}, \quad D_n(u)= \bE\suit{ \widetilde{D}_n(u)|\mathcal{F}_n},
$$
with 
$$
Z_i = \widehat{\mbu}_j^\top\mbX_i, \quad \mathcal{F}_n = \sigma \suit{(\mbX_i, Y_i)_{i=n+1}^{2n}, \suit{\mbX_i}_{i=1}^n }. 
$$
 Write
$$
\begin{aligned}
&\sqrt{k}  \set{\frac{1}{k}\sum_{i=1}^n  \widehat{\mbu}_j^\top\mbX_i \mI(U_i\le U_{k,n}) - \frac{1}{n}\sum_{i=1}^n  \widehat{\mbu}_j^\top\mbX_ig(\mbX_i^\top \bolbeta^0)  } \\
=& \sqrt{k}  \set{\frac{1}{k}\sum_{i=1}^n Z_i \mI(U_i\le U_{k,n}) - \frac{1}{n}\sum_{i=1}^n   Z_i g(\mbX_i^\top \bolbeta^0) } \\
=&	\widetilde{\Gamma}_n(  s_{n} ) +\sqrt{k}\suit{  D_n\suit{ s_n }-D_n(1) } + \sqrt{k}\suit{D_n(1)- \frac{1}{n}\sum_{i=1}^n   Z_i g(\mbX_i^\top \bolbeta^0)  } \\
=:& J_1+ J_2 +J_3,  
\end{aligned}
$$
where 
$$
s_{n} =\frac{n}{k}U_{k,n}.
$$

 By Lemma \ref{lemma:process} and Lemma \ref{lemma:order:stat}, we have that conditional on $\mathcal{F}_n,$ as $n\to\infty$, 
 $$
 J_1 = W_{1,n}(1)+o_P(1),
 $$
 where $W_{1,n}$ is defined in Lemma \ref{lemma:process}.

 Next, we handle $J_2$. Note that
$$
D_n(u) = \frac{1}{k}\sum_{i=1}^n Z_i \Pr(U_i\le ku/n|\mbX_i).
$$
By the mean-value theorem, we have that 
$$
\begin{aligned}
	\sqrt{k}(D_n(s_n)-D_n(1))= & \frac{1}{\sqrt{k}}\sum_{i=1}^n Z_i \suit{ \Pr(U_i\le ks_n/n|\mbX_i) - \Pr(U_i\le k/n|\mbX_i) } \\
	=& \frac{1}{\sqrt{k}}\sum_{i=1}^n Z_i \phi( \widetilde{U}_{in} |X_i)\frac{k}{n} \suit{s_n-1} \\
	=& \sqrt{k}(s_n-1) \frac{1}{n}\sum_{i=1}^nZ_i \phi( \widetilde{U}_{in} |\mbX_i)
\end{aligned}
$$
where $\widetilde{U}_{in}$ is between $k/n$ and $ks_n/n$. By Assumption \ref{condition:bias:gradient}, Lemma \ref{lemma:order:stat} and \eqref{bias:max}, we have that 
$$
\begin{aligned}
	J_2 =&\sqrt{k}(s_n-1) \frac{1}{n}\sum_{i=1}^nZ_i  \suit{\phi(0|\mbX_i) +B_2(\mbX_i)A_2(k/n)} \\
	=& -W_{2,n}(1) \frac{1}{n}\sum_{i=1}^nZ_i g(\mbX_i^\top \bolbeta^0) +o_p(1),
\end{aligned}
$$
 where $W_{2,n}$ is defined in Lemma \ref{lemma:process}. 
Here, the last equality follows from Assumption \ref{all:model:bias}, 
$$
\phi(0|\mbx) = \lim_{h\downarrow 0} \frac{\Pr(U\le h|\mbX=\mbx)}{h} = g(\mbx^\top \bolbeta^0). 
$$

Finally, we handle $J_3$. By Assumption \ref{all:model:bias} and \eqref{bias:max}, we have that 
	$$
	\begin{aligned}
	J_3=&\sqrt{k}\suit{ \frac{1}{k}\sum_{i=1}^n Z_i \bE\suit{\mI(U_i\le k/n)|\mbX_i}  - \frac{1}{n}\sum_{i=1}^n   Z_i g(\mbX_i^\top \bolbeta^0)  } \\
	=&\sqrt{k}\suit{  	\frac{1}{k}\sum_{i=1}^n Z_i \Pr\suit{U_i\le k/n|\mbX_i}   - \frac{1}{n}\sum_{i=1}^n   Z_i g(\mbX_i^\top \bolbeta^0)  } \\
	=&O_P(1)\sqrt{k}A(y_n)  \frac{1}{n}\sum_{i=1}^n Z_i B(\mbX_i) \\
	=& O_P(1) \sqrt{k} A(y_n) \sqrt{\log n}  = o(1). 
	\end{aligned}
	$$

 Combining the limits of $J_1, J_2$ and $J_3$, we conclude that, conditional on $\mathcal{F}_n$, as $n\to\infty$, 
 $$
 \begin{aligned}
 	 &\sqrt{k}  \set{\frac{1}{k}\sum_{i=1}^n  \widehat{\mbu}^\top\mbX_i \mI(Y_i>Y_{n-k,n}) - \frac{1}{n}\sum_{i=1}^n  \widehat{\mbu}^\top\mbX_i g(\mbX_i^\top \bolbeta^0)  }  \\
 	 =& W_{1,n}(1) - W_{2,n}(1) \frac{1}{n}\sum_{i=1}^nZ_i g(\mbX_i^\top \bolbeta^0) +o_P(1).
 \end{aligned}
 $$
 Note that 
 $$
 \begin{aligned}
 	&\text{Var}\suit{W_{1,n}(1) - W_{2,n}(1) \frac{1}{n}\sum_{i=1}^nZ_i g(\mbX_i^\top \bolbeta^0) }  \\
 	= & \frac{1}{n}\sum_{i=1}^n Z_i^2g(\mbX_i^\top\bolbeta^0 ) + \frac{1}{n}\sum_{i=1}^n g(\mbX_i^\top \bolbeta^0) \suit{\frac{1}{n}\sum_{i=1}^nZ_i g(\mbX_i^\top \bolbeta^0)}^2 - 2 \suit{\frac{1}{n}\sum_{i=1}^nZ_i g(\mbX_i^\top \bolbeta^0)}^2.
 \end{aligned}
 $$
 Since $ \frac{1}{n}\sum_{i=1}^n g(\mbX_i^\top \bolbeta^0) =1 +o(1)$, with probability tending to 1, we have that   
 $$
 \text{Var}\suit{W_{1,n}(1) - W_{2,n}(1) \frac{1}{n}\sum_{i=1}^nZ_i g(\mbX_i^\top \bolbeta^0) } = v_j^2 +o(1),
 $$
 where $v_j^2 = \widehat{\mbu}_j^\top \boldsymbol{\Sigma}_n(\bolbeta^0) \widehat{\mbu}_j $.
 Note that 
$$
\begin{aligned}
\widehat{v}_j^2 - v_j^2 = & \frac{1}{n}\sum_{i=1}^n Z_i^2 \suit{ g(\mbX_i^\top\widehat{\bolbeta}_n ) - g(\mbX_i^\top\bolbeta^0) } -\suit{ \frac{1}{n}\sum_{i=1}^n Z_i g(\mbX_i^\top\widehat{\bolbeta}_n ) }^2 + \suit{ \frac{1}{n}\sum_{i=1}^n Z_i g(\mbX_i^\top\bolbeta^0) }^2.
\end{aligned}
$$ 
Similar to the treatment of $I_3$ in the proof of Theorem \ref{theorem:normality}, we can show that, as $n\to\infty$, 
 $$
 \widehat{v}_j^2 - v_j^2=o_P(1).
 $$
  Thus, conditional on $\mathcal{F}_n$, 
 $$
 \frac{\sqrt{k}}{\widehat{v}_j} \set{\frac{1}{k}\sum_{i=1}^n \widehat{\mbu}_j^\top\mbX_i \mI(U_i\le U_{k,n}) - \frac{1}{n}\sum_{i=1}^n \widehat{\mbu}_j^\top\mbX_i g(\mbX_i^\top \bolbeta^0) } \stackrel{d}{\to} N(0, 1),
 $$
 provided that $\widehat{v}_j\ge C>0$ for some constant $C>0$. 
 Since the right-hand side of the above limit does not depend on $\mathcal{F}_n$, the above limit holds unconditionally.

We complete the proof by showing that, with probability tending to 1, $
\widehat{v}_j\ge C>0$ for some constant $C>0$. Recall that,
$$
\widehat{v}_j^2 = \frac{1}{n}\sum_{i=1}^{n} Z_i^2 \hat{g}_i - \left( \frac{1}{n}\sum_{i=1}^{n} Z_i \hat{g}_i \right)^2, 
$$
where $\hat{g}_i = g(\mathbf{X}_i^\top \widehat{\boldsymbol{\beta}}_n)$. We introduce the empirical probability measure $\boldsymbol{\pi} = \{\pi_i\}_{i=1}^n$,  
$$\pi_i = \frac{\hat{g}_i}{\sum_{i=1}^n \hat{g}_i} = \frac{\hat{g}_i}{n \bar{g}_n}, \quad \text{with } \bar{g}_n = \frac{1}{n}\sum_{i=1}^n \hat{g}_i.
$$
Since $g$ is a positive link function, we have $\pi_i \ge 0$ and $\sum_{i=1}^n \pi_i = 1$. Consequently, 
$$
\begin{aligned}
\widehat{v}_{j}^2 &= \bar{g}_n \sum_{i=1}^n Z_i^2 \pi_i - \bar{g}_n^2 \left( \sum_{i=1}^n Z_i \pi_i \right)^2 \\
&= \bar{g}_n \left[ \sum_{i=1}^n Z_i^2 \pi_i - \left( \sum_{i=1}^n Z_i \pi_i \right)^2 \right] + \bar{g}_n (1 - \bar{g}_n) \left( \sum_{i=1}^n Z_i \pi_i \right)^2 \\
&= \bar{g}_n \operatorname{Var}_{\boldsymbol{\pi}}(Z) + \bar{g}_n (1 - \bar{g}_n) \mathbb{E}_{\boldsymbol{\pi}}^2 [Z],
\end{aligned}$$
where $\operatorname{Var}_{\boldsymbol{\pi}}(Z) = \sum_{i=1}^n \left( Z_i - \mathbb{E}_{\boldsymbol{\pi}}[Z] \right)^2 \pi_i$. 
By Theorem \ref{theorem:high:main} and \eqref{eq:condition:identify}, we have that 
$\bar{g}_n \xrightarrow{P} 1$ as $n \to \infty$, rendering the second term $\bar{g}_n (1 - \bar{g}_n) \mathbb{E}_{\boldsymbol{\pi}}^2 [Z] = o_P(1)$. Thus, the strict positivity of $\widehat{v}_{j}$ asymptotically depends on showing that $\operatorname{Var}_{\boldsymbol{\pi}}(Z)$ does not vanish.

Suppose, to the contrary, that this is not true. Then there exists a
subsequence, still indexed by $n$, such that
$
\operatorname{Var}_{\boldsymbol{\pi}}(Z)
\rightarrow0
$
in probability along this subsequence.
Define
$m_n=\mathbb E_{\boldsymbol{\pi}}Z$ and $r_i = Z_i-m_n$. Then, 
$$
\operatorname{Var}_{\boldsymbol{\pi}}(Z) =\frac{1}{n\overline{g}_n } \sum_{i=1}^n \hat{g}_i r_i^2. 
$$
Since $\overline{g}_n \stackrel{P}{\to} 1$, the above convergence implies 
$$
\frac1n\sum_{i=1}^n r_i^2\hat g_i=o_P(1).
$$
By Assumption \ref{high:condition:moment:second}, we have that 
$$
\frac1n\sum_{i=1}^n r_i^2 \hat g'_i=o_P(1),
$$
where $\hat g_i' = g'(\mbX_i^\top \widehat{\bolbeta})$.

Now consider the bias-correction constraint,
$$
\left\| \frac1n
\sum_{i=1}^n \mbX_i Z_i\hat g_i'-\mathbf e_j\right\|_\infty=o_P(1).
$$
Substituting $Z_i=m_n+r_i$ gives
$$
\|m_n \frac1n \sum_{i=1}^n \mbX_i\hat g_i'+\frac1n\sum_{i=1}^n\mbX_i r_i\hat g_i'-
\mathbf e_j\|_{\infty} = o_P(1).
$$
By Assumptions \ref{condition:regular:x} and \ref{high:condition:moment}, and the Cauchy--Schwarz inequality, we have that 
$$
\begin{aligned}
\|\frac1n
\sum_{i=1}^n \mbX_i r_i\hat g_i'\|_{\infty} \le \kappa_X \frac{1}{n}\sum_{i=1}^n |r_i\hat g_i' | \le \kappa_X \suit{\frac{1}{n}\sum_{i=1}^n r_i^2 g_i'}^{1/2} \suit{\frac{1}{n}\sum_{i=1}^n g_i'}^{1/2} =o_P(1). 
\end{aligned}
$$

Hence,
$$
\|m_n\frac1n\sum_{i=1}^n\mbX_i\hat g_i'-\mathbf e_j\|_{\infty}=o_P(1).
$$

Consider the first coordinate. Since $X_{i,1}=1$,
$$
m_n
\frac1n\sum_{i=1}^n\hat g_i'=o_P(1),
$$
which implies $
m_n=o_P(1).$
However, considering the $j$-th coordinate yields
 $$
m_n\frac1n \sum_{i=1}^n X_{i,j}\hat g_i' =1+o_P(1).
$$

Since $m_n=o_P(1)$ and the empirical average is bounded in probability,
the left-hand side converges to zero, contradicting the above display.
Therefore,
$
\operatorname{Var}_{\boldsymbol{\pi}}(Z)
$
is bounded away from zero in probability. The proof is then complete.

\end{proof}

\begin{lemma}\label{lemma:process}
	Assume the same conditions as in Theorem \ref{theorem:normality}. 
	Then, conditional on $\mathcal{F}_n$, under a suitable Skorokhod construction, there exist two mean-zero Gaussian processes 
   with covariance functions 
$$
\begin{aligned}
	\text{Cov}\suit{W_{1,n}(u_1), W_{1,n}(u_2)}  = & \suit{u_1 \wedge u_2} \frac{1}{n}\sum_{i=1}^n Z_i^2 g(\mbX_i^\top\bolbeta^0), \\
	\text{Cov}\suit{W_{2,n}(u_1), W_{2,n}(u_2)}  = &\suit{u_1 \wedge u_2} \frac{1}{n}\sum_{i=1}^n  g(\mbX_i^\top\bolbeta^0), \\
	\text{Cov}\suit{W_{1,n}(u_1), W_{2,n}(u_2)}  = &\suit{u_1 \wedge u_2} \frac{1}{n}\sum_{i=1}^nZ_i  g(\mbX_i^\top\bolbeta^0),
\end{aligned}
$$
	such that
	 as $n\to\infty$, 
	$$
	\begin{aligned}
			&\sup_{0\le u \le 1} \abs{ \widetilde{\Gamma}_n(u) - W_{1,n}(u)  } = o_P(1), \\ 
			& \sup_{0\le u\le 2}\abs{ \sqrt{k}\suit{\frac{1}{k}\sum_{i=1}^n \mI\suit{U_i\le ku/n} - u\frac{1}{n}\sum_{i=1}^n  g(\mbX_i^\top \bolbeta^0) } - W_{2,n}(u) }  = o_P(1),
	\end{aligned}
	$$
	\end{lemma}
\begin{proof}[Proof of Lemma \ref{lemma:process}]
We first apply Theorem 3 of \cite{einmahl2021empirical} to establish asymptotic tightness of $\widetilde{\Gamma}_n$. 	Define
 $$
 h_{n,i}(u) = \frac{1}{\sqrt{k}} Z_i \mI\suit{ U_i\le ku/n }, \quad \mathcal{H}_n = \set{h_{n,i}(u), u\in [0,1] }. 
 $$
Note that as $n\to\infty$, 
 \begin{equation}\label{s:eq:bound:process}
  \| h_{n,i} \|_{\mathcal{H}}: = \sup_{u\in [0,1]} |h_{n,i}(u)| = |h_{n,i}(1)| = \frac{1}{\sqrt{k}} |Z_i|\mI\suit{U_i\le k/n } \le \frac{|Z_i|}{\sqrt{k}} \lesssim \frac{\sqrt{\log n}}{\sqrt{k}}\to 0 ,
 \end{equation}
 uniformly for all $1\le i\le n$. 
 Hence, 
 $$
 \lim_{n\to\infty}\sum_{i=1}^n \bE^*\suit{ \| h_{n,i} \|_{\mathcal{H}}\mI\suit{ \| h_{n,i} \|_{\mathcal{H}}> \lambda} } = 0. 
 $$

Denote $H_j = [ (j-1)\varepsilon^2, j \varepsilon^2 ] $, where $\varepsilon>0$ is a small constant. 
 By the mean-value theorem, there exist constants $\widetilde{u}_{ij} \in H_j$ such that 
  $$
  \begin{aligned}
 &\sum_{i=1}^n \bE^* \suit{ \sup_{u,v \in H_j } \abs{ h_{n,i}(u)-h_{n,i}(v) }^2 |\mathcal{F}_n }	\\
 =&\frac{1}{k} \sum_{i=1}^n  Z_i^2 \bE\msuit{\mI\suit{  k(j-1)\varepsilon^2/n<  U_i\le kj\varepsilon^2/n }|\mathcal{F}_n} \\
 =& \frac{1}{k}\sum_{i=1}^n Z_i^2  \suit{ \Pr( U_i\le k j\varepsilon^2/n |\mbX_i) - \Pr(U_i\le k (j-1)\varepsilon^2/n|\mbX_i)  }  \\
 =& \varepsilon^2  \frac{1}{n}  \sum_{i=1}^n Z_i^2  \phi(\widetilde{u}_{ij}k/n|\mbX_i ). 
 \end{aligned}
 $$
By Assumption \ref{high:condition:moment:second} and Assumption \ref{condition:bias:gradient}, we have that with probability tending to 1, 
$$
\begin{aligned}
	\frac{1}{n}  \sum_{i=1}^n Z_i^2  \phi(\widetilde{u}_{ij}k/n|\mbX_i ) = &  \frac{1}{n}  \sum_{i=1}^n Z_i^2 \suit{ \phi(0|\mbX_i) +B_2(\mbX_i)A_2(k/n)}   \\
	= &\frac{1}{n}  \sum_{i=1}^n Z_i^2  g(\mbX_i^\top \bolbeta^0) +o(1)<\infty.
\end{aligned} 
$$
Thus, we have that with probability tending to $1$, for some $c_0>0$, 
$$
 \sum_{i=1}^n \bE^* \suit{ \sup_{u,v \in H_j } \abs{ h_{n,i}(u)-h_{n,i}(v) }^2 |\mathcal{F}_n } \le c_0\varepsilon^2. 
$$
The bracket entropy $N_{\varepsilon}$ is thus of order $O(1/\varepsilon^2)$, satisfying the integral covering condition $\int_0^\delta \sqrt{\log N_{\varepsilon}} d\varepsilon < \infty$ for some $\delta>0$. Hence, $\widetilde{\Gamma}_n(u)$ is asymptotically tight. Similarly, we can establish the asymptotic tightness of the process 
$$
\Gamma_{n,2}(u) = \sqrt{k}\suit{\frac{1}{k}\sum_{i=1}^n \mI\suit{U_i\le ku/n} - u\frac{1}{n}\sum_{i=1}^n g(\mbX_i^\top \bolbeta^0)}.
$$

We now evaluate the covariance structure. For any $u, v \in [0,1]$, we have 
 \begin{align*}
 	\bE(\widetilde{\Gamma}_n(u) \widetilde{\Gamma}_n(v)|\mathcal{F}_n) = & \frac{1}{k}\sum_{i=1}^n Z_i^2 \Pr(U_i\le (u\wedge v ) k/n |\mathcal{F}_n)  =     (u\wedge v ) \frac{1}{n}\sum_{i=1}^n Z_i^2 g(\mbX_i^\top\bolbeta^0) +o(1), \\
 		\bE(\widetilde{\Gamma}_n(u) \widetilde{\Gamma}_{n,2}(v)|\mathcal{F}_n) = & \frac{1}{k}\sum_{i=1}^n Z_i \Pr(U_i\le (u\wedge v ) k/n |\mathcal{F}_n)  =     (u\wedge v ) \frac{1}{n}\sum_{i=1}^n Z_i g(\mbX_i^\top\bolbeta^0) +o(1), \\
 		\bE(\widetilde{\Gamma}_{n,2}(u) \widetilde{\Gamma}_{n,2}(v)|\mathcal{F}_n) = & \frac{1}{k}\sum_{i=1}^n \Pr(U_i\le (u\wedge v ) k/n |\mathcal{F}_n)  =     (u\wedge v ) \frac{1}{n}\sum_{i=1}^n  g(\mbX_i^\top\bolbeta^0) +o(1).
 \end{align*}
Finite-dimensional convergence to a Gaussian vector follows from the Cram\'er-Wold device and the Lindeberg-Feller CLT, where the Lindeberg condition is satisfied by the envelope bound established in \eqref{s:eq:bound:process}. This completes the proof.
 
\end{proof}

 \begin{lemma}\label{lemma:order:stat}
Assume the same conditions as in Theorem \ref{theorem:normality}. Then, conditional on $\mathcal{F}_n$, we have that as $n\to\infty$, 
 	$$
 	\sqrt{k}\suit{s_n-1} = -W_{2,n}(1) +o_P(1).
 	$$
 	Here, $W_{2,n}$ is defined in Lemma \ref{lemma:process}.
 \end{lemma}
 \begin{proof}
  
  By Lemma \ref{lemma:process} and the fact that 
  $\frac{1}{n}\sum_{i=1}^n g(\mbX_i^\top \bolbeta^0) =1+o_P(1)$, we have that 
 $$
 \sup_{0\le t\le 2}\abs{ \sqrt{k}\suit{\frac{1}{k}\sum_{i=1}^n \mI\suit{U_i\le kt/n} - t } - W_{2,n}(t) } \to 0.
 $$
 Applying Vervaat's lemma, we obtain 
 $$
 \sup_{0\le t\le 1} \abs{ \sqrt{k}\suit{\frac{n}{k}U_{kt,n}-t } + W_{2,n}(t) } =o_P(1).
 $$
 Taking $t = 1$ gives 
 $$
 	\sqrt{k}\suit{s_n-1} = -W_{2,n}(1) +o_P(1).
 $$
 The proof is then complete.
 \end{proof}

 \section{Additional Simulation Study}\label{sec:addition:simulation}
 In this section, we consider the case where $Y$ is not heavy-tailed. The response variable $Y$ is then generated according to the following mechanism:
$$
Y = 1- \frac{1}{g(\boldsymbol{X}^\top \boldsymbol{\beta}^0)}Y_0,
$$
 where the true parameter vector is specified as $$\boldsymbol{\beta}^0 = (\beta_1^0, \beta_2^0, \dots, \beta_p^0 )^\top = (\beta_1^0, 1, 0.5, -1, -0.5, 0, \dots, 0)^\top.$$ 
Here, $\beta_1^0$ is chosen to make $\bE g(\boldsymbol{X}^\top \boldsymbol{\beta}^0 ) = 1$, and $Y_0$ is a standard uniform random variable independent of $\mbX$. In this case, we have that $Y|\mbX=\mbx\sim \text{Unif}(1-1/g(\mbx^\top \bolbeta^0), 1)$ and hence 
$$
\lim_{y\to 1}\frac{1-F_Y(y|\mbX=\mbx)}{1-F_Y(y)} = g(\mbx^\top \bolbeta^0). 
$$
Other settings are the same as in Section \ref{sec:simulation}.

The finite-sample performances of the oracle estimator and the $\ell_1$-penalized estimator are summarized in Tables \ref{table:oracle:unif} and \ref{tab:simulation:unif}, respectively. We also follow the same strategy to conduct the debiasing procedure as in Section \ref{sec:simu:ci}. The finite-sample performance of the debiased estimator is reported in Table \ref{tab:inference:unif}, and the QQ plots of the studentized statistics are shown in Figure \ref{figure:unif:qq}. 
 We draw conclusions similar to those  in Section \ref{sec:simulation}. 

\begin{table}
\caption{Simulation results of the oracle estimator for the short-tailed case.}
\label{table:oracle:unif}
\centering
 \begin{tabular}[t]{lrr}
\toprule
g & Bias ($\beta_2$) & SD ($\beta_2$)\\
\midrule
exp  & -0.036 & 0.153\\
softplus & -0.007 & 0.323\\
\bottomrule
\end{tabular}
\end{table}

 \begin{table}[htbp]
  \centering
  \caption{Simulation results of the $\ell_1$-penalized estimator for the short-tailed case.}
  \label{tab:simulation:unif}
  \footnotesize
  \begin{tabular}{lrrrrrrr}
\toprule
g & p & $c_0$ & accuracy & bias($\beta_2$) & sd ($\beta_2$) & bias($\beta_{10}$) & sd ($\beta_{10}$)\\
\midrule
exp & 50 & 0.5 & 0.799 & -0.149 & 0.147 & 0.000 & 0.041\\
exp & 50 & 1.0 & 0.953 & -0.290 & 0.137 & -0.001 & 0.011\\
exp & 100 & 0.5 & 0.824 & -0.161 & 0.148 & 0.000 & 0.035\\
exp & 100 & 1.0 & 0.970 & -0.325 & 0.140 & -0.001 & 0.009\\
softplus & 50 & 0.5 & 0.767 & -0.322 & 0.262 & -0.001 & 0.076\\
softplus & 50 & 1.0 & 0.931 & -0.664 & 0.204 & -0.001 & 0.019\\
softplus & 100 & 0.5 & 0.795 & -0.345 & 0.265 & -0.003 & 0.067\\
softplus & 100 & 1.0 & 0.957 & -0.714 & 0.198 & -0.001 & 0.019\\
\bottomrule
\end{tabular}
\end{table}

 \begin{table}[htbp]
  \centering
  \caption{Simulation results of the debiased estimator for the short-tailed case.}
  \label{tab:inference:unif}
  \footnotesize
 	\begin{tabular}{lrrrrrrr}
\toprule
g & p & coverage($\beta_2$) & bias($\beta_2$) & sd($\beta_2$) & coverage($\beta_{10}$)& bias($\beta_{10}$) &sd($\beta_{10}$)\\
\midrule
exp & 50 & 0.967 & -0.019 & 0.153 & 0.965 & 0.000 & 0.183\\
exp & 100 & 0.959 & -0.016 & 0.158 & 0.962 & 0.003 & 0.185\\
softplus & 50 & 0.951 & -0.063 & 0.298 & 0.950 & -0.002 & 0.342\\
softplus & 100 & 0.948 & -0.052 & 0.302 & 0.956 & -0.002 & 0.338\\
\bottomrule
\end{tabular}
 \end{table}

 \begin{figure}[htbp]
    \centering
    \begin{subfigure}{\textwidth}
        \centering
       \includegraphics[width=0.9\textwidth]{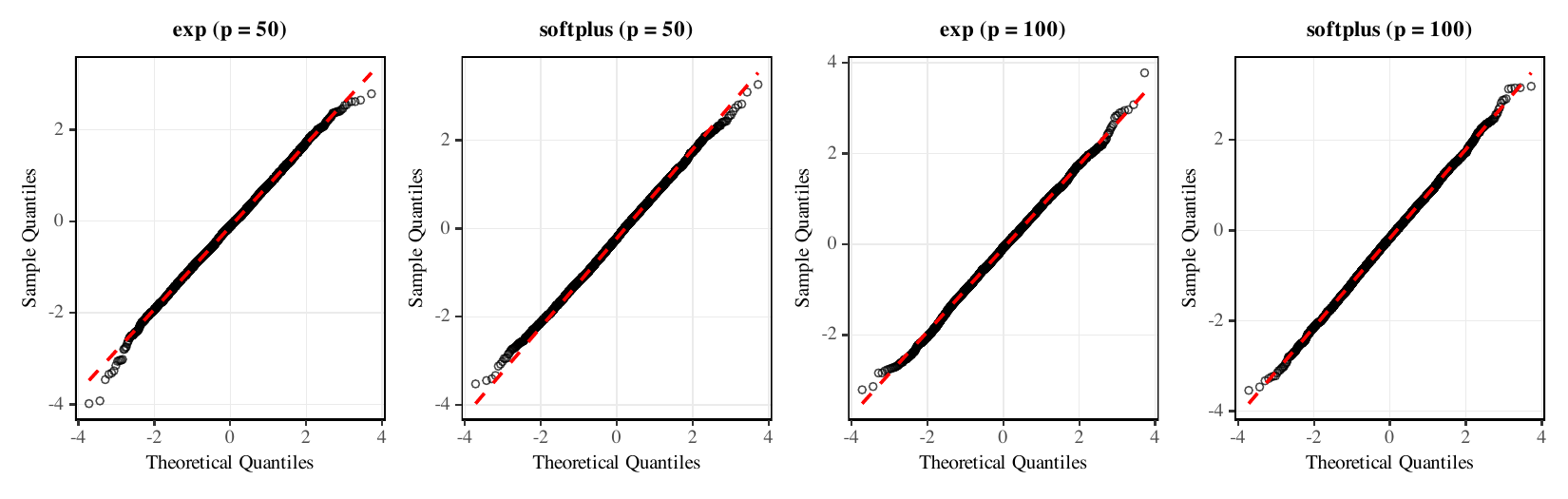}
        \caption{QQ plots for coefficient $\beta_2$.}
    \end{subfigure}    
    \vspace{1em} 
    \begin{subfigure}{\textwidth}
        \centering
        \includegraphics[width=0.9\textwidth]{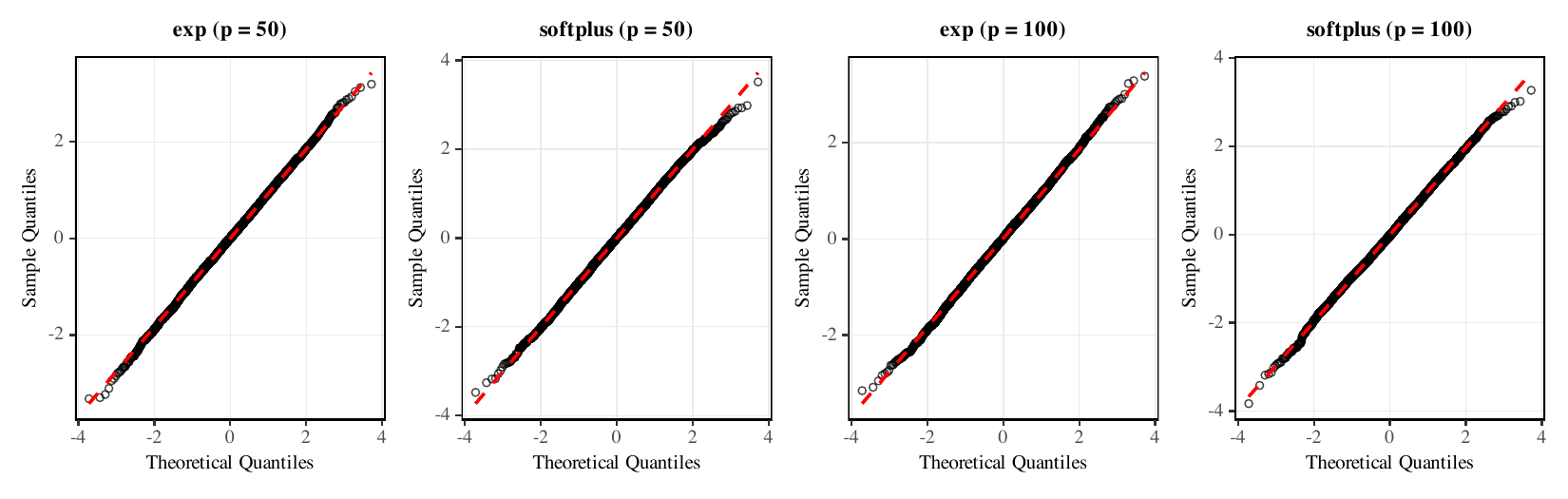}
        \caption{QQ plots for coefficient $\beta_{10}$.}
    \end{subfigure}
    \caption{QQ plots for the debiased estimator for the short-tailed case.}
    \label{figure:unif:qq}
\end{figure}
\clearpage

\end{document}